\documentclass[conference]{IEEEtran}

\IEEEoverridecommandlockouts

\usepackage[T1]{fontenc}
\usepackage{cite}
\usepackage{hyperref}
\usepackage{amsmath,amssymb,amsfonts,amsthm}
\usepackage{algorithmic}
\usepackage{graphicx}
\usepackage{textcomp}
\usepackage{xcolor}
\usepackage{braket}
\usepackage{lipsum}
\usepackage{tikz}
\usepackage{quantikz}
\usepackage{float}
\usepackage{subcaption}
\usepackage{array}
\usepackage{booktabs}
\usepackage{algorithm2e}
\usepackage{xfrac}
\usepackage{mdframed}
\usepackage{adjustbox} 
\usepackage{amsmath,amssymb,amsfonts}
\usepackage{braket}
\usepackage{tabularx}

\RestyleAlgo{ruled}
\SetKwComment{Comment}{/* }{ */}

\usetikzlibrary{positioning}

\def\BibTeX{{\rm B\kern-.05em{\sc i\kern-.025em b}\kern-.08em
    T\kern-.1667em\lower.7ex\hbox{E}\kern-.125emX}}

\newtheorem{lemma}{Lemma}
\newtheorem*{remark}{Remark}
\newtheorem{definition}{Definition}
\newtheorem{prop}{Proposition}

\newtheorem{cor*}{Corollary}

\newtheorem*{resm}{Monolithic EC Fabric}
\newtheorem*{resmd}{Modular EC Fabric}
\newcommand{\eqdef}{\stackrel{\triangle}{=}}
\newtheorem{des}{Design Principle}
\begin{document}

\title{Quantum Routers: A Switching-Fabric Framework for Quantum-Native Forwarding}

\author{ Jessica Illiano~\IEEEmembership{Member,~IEEE}, Caterina De Risi, \\Angela Sara Cacciapuoti~\IEEEmembership{Senior~Member,~IEEE}, and Marcello Caleffi~\IEEEmembership{Senior~Member,~IEEE} 
    \thanks{The authors are with the \href{www.quantuminternet.it}{www.QuantumInternet.it} research group, University of Naples Federico II, Naples, 80125 Italy.}
    \thanks{This work has been funded by the European Union under Horizon Europe ERC-CoG grant QNattyNet, n.101169850. Views and opinions expressed are however those of the author(s) only and do not necessarily reflect those of the European Union or the European Research Council Executive Agency. Neither the European Union nor the granting authority can be held responsible for them. A preliminary conference version of this work has been accepted in the Proceedings of IEEE ICC’26 \cite{IllDerCal-26}}
}
\maketitle

\begin{abstract}
Forwarding in quantum networks cannot be realized by directly transposing classical switching fabrics, since the no-cloning theorem and the quantum measurement postulate constrain the direct relay of quantum information while ruling out copy-based buffering and inspection. In this paper, we propose a \textit{switching-fabric framework} for quantum routers based on multipartite entanglement. Specifically, we formalize the notion of an \textit{entanglement-based switching fabric}, in which a graph state acts as the forwarding resource and entanglement forwarding is realized through local Pauli measurements. We translate the classical notions of blocking and non-blocking operation into structural conditions for entanglement-based fabrics, by deriving the \textit{edge-controlled (EC) design principle} for non-blocking operation. We instantiate this principle through a monolithic \textit{EC crossbar} and a modular \textit{Clos-type EC fabric}, for which we characterize resource scaling and identify the regime where the modular design becomes more resource-efficient than the monolithic one. Finally, a forwarding-latency analysis establishes a fundamental distinction between \textit{matching-oblivious} and \textit{matching-driven} forwarding: the proposed EC fabrics realize all requested input-output entanglement links with constant forwarding depth under sufficient measurement parallelism, whereas matching-driven EPR-based fabrics exhibit latency that scales with the number of requested connections. The proposed framework provides a hardware-agnostic foundation for quantum-router switching fabrics.
\end{abstract}


%-----------------------
%-----Sec.I-------------
%-----------------------
\section{Introduction}
\label{sec:1}

At the heart of a classical router lies the switching fabric, i.e., the internal component responsible for forwarding packets from input-ports to output-ports according to routing decisions. Such fabrics, including crossbar and Clos-type architectures, are designed to support multiple input-output connections in parallel and, when non-blocking, to allow any idle input-port to be connected to any idle output-port.

The above cannot be directly transposed to the quantum domain, since the no-cloning theorem and the quantum measurement postulate constrain the direct relay of quantum information while ruling out copy-based buffering and inspection. Hence, the quantum counterpart of packet forwarding cannot be based on moving, copying, and inspecting information carriers across intermediate nodes. Instead, the Quantum Internet is built on entanglement as a fundamental network resource: the network does not forward information itself, but distributes, maintains, and manipulates entangled resources that can later be consumed by quantum applications~\cite{CalCac-25,CacCal-26}. In this sense, entanglement qubits (ebits) play the role of network-level delivery units, and forwarding becomes the operation of steering entanglement resources across the network rather than relaying bit-packets or unknown quantum states hop by hop~\cite{CalCac-25,CacCalIll-25,AbaCubMai-25}. Following~\cite{CalCac-25}, we refer to this paradigm as \textit{entanglement forwarding}. Building on these considerations, in this paper we design a \textit{switching-fabric framework} for quantum routers. Specifically, we investigate how multipartite entanglement states can be engineered to function as internal forwarding resources, enabling entanglement forwarding between input- and output-ports through local Pauli measurements, independently of any specific physical implementation. The proposed approach fundamentally differs from hardware-driven photonic or optomechanical routers: rather than relying on physical wiring, it activates entanglement links adaptively through measurements. Our contributions are as follows:
\begin{itemize}
    \item We formalize the notion of an \textit{entanglement-based switching fabric} for quantum routers, in which entanglement forwarding is realized through local Pauli measurements on a multipartite entanglement resource.

    \item We translate the classical notions of blocking and non-blocking operation into structural conditions for entanglement-based switching fabrics, and derive the design principle required to achieve a \textit{non-blocking} quantum fabric, namely the \textit{edge-controlled (EC) design principle}.

    \item We instantiate the EC design through a monolithic \textit{EC crossbar} and a modular \textit{Clos-type EC fabric}. For both architectures, we derive the resource scaling and we identify the threshold above which the modular architecture becomes more resource-efficient than the monolithic one.

    \item We provide a forwarding-latency and resource comparison between the proposed EC fabrics and other relevant architectures based on EPR-resources. The analysis identifies two forwarding regimes: \textit{matching-oblivious forwarding}, where the resource is provisioned to support any admissible matching before the matching is assigned; and \textit{matching-driven forwarding}, where the requested matching determines the forwarding operations to be performed, thereby introducing matching-dependent latency.
\end{itemize}

The proposed framework provides a hardware-agnostic foundation for quantum-router switching fabrics.
%---------
\subsection{Related Work}
\label{sec:1.1}

In quantum networks, several classes of devices have been proposed to support node interconnection and entanglement distribution. However, since quantum-network architectures are still at an early stage of conceptual development, consistent terminology and clear functional definitions for the quantum counterparts of classical devices, such as repeaters, switches, and routers, are still lacking. Quantum repeaters were originally conceived to extend the range of entanglement distribution through entanglement swapping, with later generations including entanglement purification and quantum error correction~\cite{MurLiKim-16}. Quantum switches generalize this functionality to multi-user settings, whereas quantum routers extend it to multi-hop scenarios, enabling end-to-end entanglement distribution among distant nodes~\cite{AzuTamLo-15,BenHajVan-23,Alvaro-24,LemCer-13,ShiMan-23,Promponas-23,gauthier-23}.

Most of these works follow a bipartite-resource paradigm: routers, switches, and repeaters are modeled as intermediate devices supporting end-to-end entanglement distribution by generating or manipulating EPR pairs through BSMs. Their operational mode is thus constrained by point-to-point entanglement availability, quantum memory capacity, fidelity, and entanglement-generation rate. Consequently, most investigations analyze these devices as components of larger networked systems, with the device itself abstracted as a functional entity executing BSM operations or scheduling EPR generation according to a routing metric or optimization algorithms. Contributions that explicitly examine device-level functionality are instead primarily hardware-oriented~\cite{AzaMohPol-24,KuaJin-19,CumJuaPue-25}.

Beyond bipartite-resource architectures, another line of research investigates pre-generated multipartite entanglement as a resource for network connectivity and functionalities \cite{MazPelCac-26}. In~\cite{AzuTamLo-15}, multipartite entanglement is used to embed purification within a repeater protocol. In~\cite{MazCalCac-24-1}, multipartite entanglement is exploited to engineer on-demand entanglement-overlay graphs within a Quantum Local Area Network (QLAN), where local Pauli measurements at a central orchestrator establish entanglement links among client nodes beyond the physical graph constraints. More recently, \cite{RamMorDur-26} proposes a distributed entanglement-based architecture in which Bell pairs arranged according to a hypercube structure enable universal pairwise connectivity via entanglement swapping along edge-disjoint paths, achieving $O(n\log n)$ total resource scaling. 

These works demonstrate the potential of pre-generated entangled resources for network connectivity either with a centralized approach \cite{MazCalCac-24-1}, or in a distributed manner \cite{RamMorDur-26}. Differently, our proposal focuses on a multipartite entangled resource localized within a single router node and used as its internal switching fabric. This device-level perspective raises design questions that remain unaddressed in the existing literature: what makes an entangled resource a valid quantum switching fabric, when such a fabric is blocking or non-blocking, and how modular non-blocking fabrics can be constructed.
In this paper, we adopt the term \textit{quantum router} to emphasize the directional input-output forwarding functionality of the proposed device, in analogy with classical routers. This terminology also reflects the intended role of the proposed fabric within network-level entanglement forwarding, rather than only local-area interconnection.

%---------------------
%-----Sec. II---------
%---------------------
\section{Background}
\label{sec:2}

Here, we outline the preliminaries that are used throughout the remainder of the paper. 

Graph states constitute a versatile class of multipartite entangled states for quantum networking~\cite{PirDur-21, FrePirDur-24,CheIllCac-24,MazCalCac-24}. An $n$-qubit graph state has the useful property of admitting a one-to-one correspondence with a graph representation~\cite{HeiDurEis-06}. Accordingly, a graph state $\ket{G}$ is univocally associated with a graph $G=(V,E)$, where each vertex $a\in V$, with $|V|=n$, represents a qubit, and each edge $\{a,b\}\in E$ denotes a controlled-$Z$ operation between qubits $a$ and $b$.
Indeed, owing to this one-to-one mapping, single-qubit Pauli measurements, $\{\sigma_{\chi}\}$ with $\chi \in \{x,y,z\}$, performed on the qubits of a graph state can be translated into the corresponding graph-level operations. The post-measurement graph is obtained -- up to local unitaries -- by applying the graph transformations as described in~\cite{HeiEisBri-04,HeiDurEis-06}, namely, vertex deletion for $\sigma_z$, local complementation for $\sigma_y$, and nested local complementations for $\sigma_x$ measurements. The detailed mathematical formulations and derivations can be found in~\cite{HeiDurEis-06,MazCalCac-24,CheIllCac-24}.

\begin{definition}[Measurement strategy]
\label{def:01}
Let $A$ be a finite index set of qubits in $V$. 
A measurement strategy $\gamma(\cdot)$ is a function that assigns to each qubit $i\in A$ a label $\chi_i \in\{x,y,z\}$, indicating which single-qubit Pauli measurement, $\sigma^{i}_{\chi_i}$ must be performed on $i$: 
\begin{equation}  
\label{eq:01}
\gamma(i)=\sigma_{\chi_i}^{(i)},  \quad \forall i \in A,  \, \textrm{with}    \,\chi_i \in \{x,y,z\}.
\end{equation}
\end{definition}
Accordingly, given a strategy $\gamma(\cdot)$, $\mathcal{P}^{A}_{\gamma}\; \eqdef\; \bigotimes_{i\in A}\, \sigma_{\chi_i}^{(i)}$ denotes the corresponding product of Pauli-measurement choices acting on $A$.
If the global system is defined on a larger set of qubits, $V\supseteq A$, we use the natural embedding so that $\mathcal{P}^{V}_{\gamma}$ leaves all qubits in $V\setminus A$ unmeasured\footnote{In~\eqref{eq:02}, we used the common notation abuse of omitting the explicit tensor ordering for readability. The ordering in the tensor product remains nevertheless meaningful.}: 
\begin{equation}
\label{eq:02}
\mathcal{P}^{V}_{\gamma}\;\; \eqdef\;\; 
\left( \bigotimes_{i\in A}\, \sigma_{\chi_i}^{(i)} \right)
\;\otimes\;
\left( \bigotimes_{j\in V\setminus A} \! I^{(j)} \right).
\end{equation}
%--------
%--------
\begin{definition}[Achievability]
\label{def:02}
Let $\ket{G}$ be a graph state associated with the graph $G=(V,E)$, and let $G'=(V',E')$ denote the targeted graph specified by a desired connection pattern. 
We say that $G'$ is achievable through a measurement strategy $\gamma(\cdot)$ if there exists a subset of qubits $A\subseteq V$ such that: 
\begin{equation}
\label{eq:03}
\mathcal{P}^V_{\gamma}\ket{G} \;\overset{\text{LU}}{=}\; 
    \ket{G'} \otimes \ket{\phi_{\mathrm{idle}}},
\end{equation}
where $\mathcal{P}^V_{\gamma}$ is defined in~\eqref{eq:02}, $\overset{\mathrm{LU}}{=}$ denotes the equivalence up to local unitaries, $\ket{G'}$ is the post-measurement state corresponding to the graph $G'$, and $\ket{\phi_{\mathrm{idle}}}$ denotes the residual state on the subsystem of qubits $V \setminus (V' \cup A)$.
\end{definition}
\noindent According to Def.~\ref{def:02}, a target state $\ket{G'}$ is decoupled from the residual state defined on the subsystem of qubits $V \setminus (V' \cup A)$ via the measurement strategy. 
The vertices in $A$ on which the measurement strategy $\gamma$ acts through Pauli measurements are removed from the original graph. Hence, each achievable graph $G'$ can be associated with a specific \textit{depleted set} of vertices, those measured to transform the initial graph into the desired target graph $G'$. This concept is formalized as follows.
\begin{definition}[Depleted set]
\label{def:03}
Given a graph state $\ket{G}$ associated with a graph $G=(V,E)$ and a measurement strategy $\gamma(\cdot)$, the \textit{depleted set} $\Gamma(G')\subseteq V$ corresponding to an achievable graph $G'$ is defined as $\Gamma(G')\eqdef \{\, i \in V \;|\; \gamma(i)= \sigma^{(i)}_{\chi_i}, \text{\rm with}\,\, \chi_i \in \{x,y,z\}\,\}$. The set $\Gamma(G')$ thus contains the vertices measured under $\gamma(\cdot)$ to obtain the connection pattern described by $G'$. 
\end{definition}

%------------------------------------
%------Sec. III ---------------------
%------------------------------------
\section{Monolithic Entanglement-based Switching Fabric Design}
\label{sec:3}
As discussed in Sec.~\ref{sec:1}, the objective of a quantum router is to dynamically interconnect any set of disjoint pairs of \textit{input-port qubit} and \textit{output-port qubit}, thereby realizing the forwarding functionality in the quantum domain. We note that we adopted the classical router terminology for the sake of clarity.
To this aim, a graph state $\ket{G}$ is used as the fundamental resource underpinning the operation of the switching fabric. Accordingly, we define $I=\{i_1,\dots,i_N\}$ and $O=\{o_1,\dots,o_N\}$ as the sets of input- and output-port qubits, respectively, which are disjoint subsets of the vertex set, i.e., $I \cup O = V$ and $I\cap O=\emptyset$. With this in mind, the achievable graph $G'$ introduced in Def.~\ref{def:02} corresponds to the graph obtained by properly manipulating the original graph $G$, to activate a set of the aforementioned disjoint interconnections between input- and output-port qubits, as formally captured by the following definition.
%-----
\begin{definition}[Matching Set]
\label{def:04}
    Let $\ket{G}$ be a graph state associated with a graph $G=(V,E)$, where $V=I \cup O$ is partitioned into two disjoint subsets of port qubits: the input-ports $I={i_1,\dots,i_N}$ and the output-ports $O={o_1,\dots,o_N}$. The \textit{matching set} $E' \subset I \times O$, with $|E'|\leq N$ represents the desired pattern of entanglement links established between disjoint input- and output-port pairs, as a result of a measurement strategy $\gamma(\cdot)$ applied to the resource state $\ket{G}$, i.e.:
\begin{equation}
    \label{eq:05}
    \ket{G'}\overset{\mathrm{LU}}{=}\bigotimes_{(i,o)\in E'}\ket{\Phi_{io}},
\end{equation}
%where $\overset{\mathrm{LU}}{=}$ denotes the equivalence up to local unitaries and 
where $\ket{\Phi_{io}}$ is a Bell state shared between $i$ and $o$. The corresponding achieved graph $G'=(V',E')$ is:
\begin{align}
\label{eq:04}
G'&=(V',E'), \, \text{\rm with}\\&\nonumber
    V' =\{i,o \;|\; (i,o) \in E'\}\subseteq I \cup O,\\&\nonumber
    E'=\{(i_i,o_j),(i_k,o_r) \;|\;i_i,i_k\in I \wedge o_j,o_r \in O,\\&\nonumber  \text{\rm with} \, \, i_i\neq i_k, o_j \neq o_r \}.
\end{align}
\end{definition}

According to Def.~\ref{def:04}, the edges in $E'$ specify the dynamic connection pattern, or
\textit{forwarding configuration}, realized by the quantum router in a single operational round. The vertices in $V'$ correspond to the \textit{active ports} participating in the current matching, while all remaining port qubits in $(I \cup O)\setminus V'$ are referred to as \textit{idle ports}.
 In the following, for notational simplicity and with a slight notation abuse, we will use the symbols $G'$, $V'$, and $E'$ interchangeably when referring to the achieved configuration. In particular, we introduce the unified symbol $\mathcal{M}$ to denote either $G'$, $V'$, or $E'$, with the intended meaning clear from the context.

\begin{definition}[$N$-Crossbar resource]
    \label{def:05}
    Let $\ket{G}$ be the graph state associated with the graph $G=(V,E)$, whose vertex set is partitioned into the input- and output-ports, i.e., $V=I\cup O$, with $I\cap O=\emptyset$ and $|I|=|O|=N$. Then $\ket{G}$ is a valid $N$-crossbar resource if, for every matching $\mathcal{M}\subset I \times O$ as defined in Def.~\ref{def:04}, with $|\mathcal{M}|\leq N$, there exists a measurement strategy $\gamma(\cdot)$ acting on the qubits in $V$ such that:
    \begin{equation} 
    \label{eq:06}
        \mathcal{P}^V_\gamma\ket{G} \;\overset{\mathrm{LU}}{=}\; 
      \ket{G'}\ \otimes\ \ket{\phi_{\text{idle}}},
    \end{equation}
    where $G'=(V',E')$ is the achieved graph associated with the requested matching configuration $\mathcal{M}$, so that $\ket{G'}$ consists of the disjoint input-output Bell pairs as in \eqref{eq:05}. $\ket{\phi_{\mathrm{idle}}}$ denotes the residual state on the subsystem $V\setminus (V'\cup \Gamma(G'))$, with $\Gamma(G')$ the depleted set induced by $\gamma(\cdot)$ as defined in Def.~\ref{def:03}.  
\end{definition}

As represented in the leftmost part of Fig.~\ref{fig:01}, the simplest example of a crossbar resource is a 3-qubit GHZ state, where $i_1$ acts as an input qubit and $o_1$, $o_2$ as output qubits. The input-output connection between the pair $(i_1,o_j)$ is realized through a Pauli measurement on the remaining output-port qubit $o_r, r\neq j$. In principle, this elementary resource could be scaled, as represented in the rightmost part of Fig.~\ref{fig:01}, where the sets of input- and output-port qubits are interconnected through a graph state associated with a complete bipartite graph. However, this construction can realize only single-pair configurations, i.e., matchings with $|\mathcal{M}|=1$, under any admissible measurement strategy within the considered Pauli-measurement model. Consequently, GHZ-equivalent resources do not provide a scalable switching fabric, as they cannot support simultaneous independent input-output connections.

To design scalable entanglement-based switching fabrics, it is therefore necessary to reinterpret the notion of blocking and non-blocking operational modes from classical switching theory in the quantum domain, where forwarding is realized through local operations and measurements rather than wired electronic components. Classically, a switching fabric with $N$ inputs and $N$ outputs is said to be \textit{non-blocking} if any idle input can always be connected to any idle output\footnote{Different formulations of the non-blocking property exist \cite{Clos-53,Ben-62,Kol-98}. We adopted the definition introduced in the early work on switching networks \cite{Clos-53}.}, where an idle input or output is a port not currently engaged in any active connection within the fabric. With this in mind, we formulate Lemma~\ref{lem:01}. 
%--------
\begin{figure}
    \centering
    \resizebox{0.85\columnwidth}{!}{
    \input{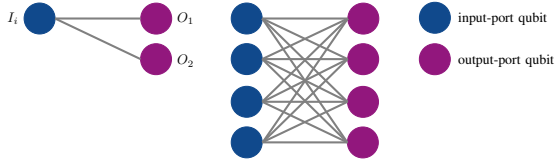}
    }
    \caption{ \footnotesize Graph representation of a GHZ-equivalent port-only resource, as an entanglement-based switching fabric. Pauli measurements on port qubits can realize only one input-output connection per operational round.}
    \hrulefill
    \label{fig:01}
\end{figure}
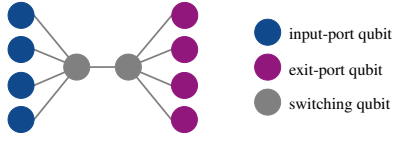
\begin{figure}
\centering
    \resizebox{0.6\columnwidth}{!}{
    \tikzset{every picture/.style={line width=0.75pt}} %set default line width to 0.75pt        

\begin{tikzpicture}[x=0.75pt,y=0.75pt,yscale=-1,xscale=1]
%uncomment if require: \path (0,455); %set diagram left start at 0, and has height of 455

%Shape: Circle [id:dp4006638057169051] 
\draw  [color={rgb, 255:red, 146; green, 23; blue, 125 }  ,draw opacity=1 ][fill={rgb, 255:red, 146; green, 23; blue, 125 }  ,fill opacity=1 ] (408.85,113.64) .. controls (408.82,121.84) and (402.15,128.47) .. (393.94,128.44) .. controls (385.74,128.41) and (379.11,121.74) .. (379.14,113.53) .. controls (379.17,105.33) and (385.85,98.7) .. (394.05,98.73) .. controls (402.26,98.76) and (408.88,105.44) .. (408.85,113.64) -- cycle ;
%Shape: Circle [id:dp2731331602898256] 
\draw  [color={rgb, 255:red, 16; green, 73; blue, 140 }  ,draw opacity=1 ][fill={rgb, 255:red, 16; green, 73; blue, 140 }  ,fill opacity=1 ] (379.2,73.74) .. controls (379.2,65.53) and (385.85,58.88) .. (394.05,58.88) .. controls (402.25,58.88) and (408.9,65.53) .. (408.91,73.73) .. controls (408.91,81.93) and (402.26,88.58) .. (394.06,88.59) .. controls (385.85,88.59) and (379.2,81.94) .. (379.2,73.74) -- cycle ;
%Shape: Circle [id:dp3467256136776755] 
\draw  [color={rgb, 255:red, 128; green, 128; blue, 128 }  ,draw opacity=1 ][fill={rgb, 255:red, 128; green, 128; blue, 128 }  ,fill opacity=1 ] (190.69,113.59) .. controls (190.66,121.79) and (183.99,128.42) .. (175.78,128.39) .. controls (167.58,128.35) and (160.95,121.68) .. (160.98,113.48) .. controls (161.01,105.27) and (167.69,98.65) .. (175.89,98.68) .. controls (184.1,98.71) and (190.72,105.38) .. (190.69,113.59) -- cycle ;
%Shape: Circle [id:dp1211509143998909] 
\draw  [color={rgb, 255:red, 128; green, 128; blue, 128 }  ,draw opacity=1 ][fill={rgb, 255:red, 128; green, 128; blue, 128 }  ,fill opacity=1 ] (249.69,113.59) .. controls (249.66,121.79) and (242.99,128.42) .. (234.78,128.39) .. controls (226.58,128.35) and (219.95,121.68) .. (219.98,113.48) .. controls (220.01,105.27) and (226.69,98.65) .. (234.89,98.68) .. controls (243.1,98.71) and (249.72,105.38) .. (249.69,113.59) -- cycle ;
%Straight Lines [id:da117599859790543] 
\draw [color={rgb, 255:red, 128; green, 128; blue, 128 }  ,draw opacity=1 ][fill={rgb, 255:red, 128; green, 128; blue, 128 }  ,fill opacity=1 ][line width=1.5]    (175.84,113.53) -- (127.11,173.45) ;
%Straight Lines [id:da7789273384891272] 
\draw [color={rgb, 255:red, 128; green, 128; blue, 128 }  ,draw opacity=1 ][fill={rgb, 255:red, 128; green, 128; blue, 128 }  ,fill opacity=1 ][line width=1.5]    (175.84,113.53) -- (127.11,134.45) ;
%Straight Lines [id:da9817147481165613] 
\draw [color={rgb, 255:red, 128; green, 128; blue, 128 }  ,draw opacity=1 ][fill={rgb, 255:red, 128; green, 128; blue, 128 }  ,fill opacity=1 ][line width=1.5]    (175.84,113.53) -- (127.11,93.45) ;
%Straight Lines [id:da3550945038271257] 
\draw [color={rgb, 255:red, 128; green, 128; blue, 128 }  ,draw opacity=1 ][fill={rgb, 255:red, 128; green, 128; blue, 128 }  ,fill opacity=1 ][line width=1.5]    (175.84,113.53) -- (127.11,54.45) ;
%Straight Lines [id:da8306249996140778] 
\draw [color={rgb, 255:red, 128; green, 128; blue, 128 }  ,draw opacity=1 ][fill={rgb, 255:red, 128; green, 128; blue, 128 }  ,fill opacity=1 ][line width=1.5]    (283.98,173.48) -- (234.84,113.53) ;
%Straight Lines [id:da7332712022620017] 
\draw [color={rgb, 255:red, 128; green, 128; blue, 128 }  ,draw opacity=1 ][fill={rgb, 255:red, 128; green, 128; blue, 128 }  ,fill opacity=1 ][line width=1.5]    (284.13,134.48) -- (234.84,113.53) ;
%Straight Lines [id:da586505833685375] 
\draw [color={rgb, 255:red, 128; green, 128; blue, 128 }  ,draw opacity=1 ][fill={rgb, 255:red, 128; green, 128; blue, 128 }  ,fill opacity=1 ][line width=1.5]    (284.29,93.48) -- (234.84,113.53) ;
%Straight Lines [id:da9849264401931447] 
\draw [color={rgb, 255:red, 128; green, 128; blue, 128 }  ,draw opacity=1 ][fill={rgb, 255:red, 128; green, 128; blue, 128 }  ,fill opacity=1 ][line width=1.5]    (284.44,54.48) -- (234.84,113.53) ;
%Straight Lines [id:da8011649032961451] 
\draw [color={rgb, 255:red, 128; green, 128; blue, 128 }  ,draw opacity=1 ][fill={rgb, 255:red, 128; green, 128; blue, 128 }  ,fill opacity=1 ][line width=1.5]    (219.98,113.48) -- (190.69,113.59) ;
%Shape: Circle [id:dp9569043567774285] 
\draw  [color={rgb, 255:red, 128; green, 128; blue, 128 }  ,draw opacity=1 ][fill={rgb, 255:red, 128; green, 128; blue, 128 }  ,fill opacity=1 ] (408.85,153.64) .. controls (408.82,161.84) and (402.15,168.47) .. (393.94,168.44) .. controls (385.74,168.41) and (379.11,161.74) .. (379.14,153.53) .. controls (379.17,145.33) and (385.85,138.7) .. (394.05,138.73) .. controls (402.26,138.76) and (408.88,145.44) .. (408.85,153.64) -- cycle ;
%Shape: Circle [id:dp48259057918474535] 
\draw  [color={rgb, 255:red, 146; green, 23; blue, 125 }  ,draw opacity=1 ][fill={rgb, 255:red, 146; green, 23; blue, 125 }  ,fill opacity=1 ] (313.69,173.59) .. controls (313.66,181.79) and (306.99,188.42) .. (298.78,188.39) .. controls (290.58,188.35) and (283.95,181.68) .. (283.98,173.48) .. controls (284.01,165.27) and (290.69,158.65) .. (298.89,158.68) .. controls (307.1,158.71) and (313.72,165.38) .. (313.69,173.59) -- cycle ;
%Shape: Circle [id:dp5424459554130382] 
\draw  [color={rgb, 255:red, 146; green, 23; blue, 125 }  ,draw opacity=1 ][fill={rgb, 255:red, 146; green, 23; blue, 125 }  ,fill opacity=1 ] (313.84,134.59) .. controls (313.81,142.79) and (307.14,149.42) .. (298.93,149.39) .. controls (290.73,149.36) and (284.1,142.68) .. (284.13,134.48) .. controls (284.16,126.27) and (290.84,119.65) .. (299.04,119.68) .. controls (307.24,119.71) and (313.87,126.38) .. (313.84,134.59) -- cycle ;
%Shape: Circle [id:dp9153352468347117] 
\draw  [color={rgb, 255:red, 146; green, 23; blue, 125 }  ,draw opacity=1 ][fill={rgb, 255:red, 146; green, 23; blue, 125 }  ,fill opacity=1 ] (314,93.59) .. controls (313.97,101.79) and (307.29,108.42) .. (299.09,108.39) .. controls (290.89,108.36) and (284.26,101.68) .. (284.29,93.48) .. controls (284.32,85.27) and (291,78.65) .. (299.2,78.68) .. controls (307.4,78.71) and (314.03,85.38) .. (314,93.59) -- cycle ;
%Shape: Circle [id:dp2768413139295196] 
\draw  [color={rgb, 255:red, 146; green, 23; blue, 125 }  ,draw opacity=1 ][fill={rgb, 255:red, 146; green, 23; blue, 125 }  ,fill opacity=1 ] (314.15,54.59) .. controls (314.12,62.79) and (307.44,69.42) .. (299.24,69.39) .. controls (291.04,69.36) and (284.41,62.68) .. (284.44,54.48) .. controls (284.47,46.27) and (291.15,39.65) .. (299.35,39.68) .. controls (307.55,39.71) and (314.18,46.38) .. (314.15,54.59) -- cycle ;
%Shape: Circle [id:dp39721861658296587] 
\draw  [color={rgb, 255:red, 16; green, 73; blue, 140 }  ,draw opacity=1 ][fill={rgb, 255:red, 16; green, 73; blue, 140 }  ,fill opacity=1 ] (97.4,54.46) .. controls (97.4,46.26) and (104.05,39.61) .. (112.26,39.6) .. controls (120.46,39.6) and (127.11,46.25) .. (127.11,54.45) .. controls (127.11,62.66) and (120.46,69.31) .. (112.26,69.31) .. controls (104.06,69.31) and (97.41,62.66) .. (97.4,54.46) -- cycle ;
%Shape: Circle [id:dp8346967460321543] 
\draw  [color={rgb, 255:red, 16; green, 73; blue, 140 }  ,draw opacity=1 ][fill={rgb, 255:red, 16; green, 73; blue, 140 }  ,fill opacity=1 ] (97.4,93.46) .. controls (97.4,85.26) and (104.05,78.61) .. (112.26,78.6) .. controls (120.46,78.6) and (127.11,85.25) .. (127.11,93.45) .. controls (127.11,101.66) and (120.46,108.31) .. (112.26,108.31) .. controls (104.06,108.31) and (97.41,101.66) .. (97.4,93.46) -- cycle ;
%Shape: Circle [id:dp9190396670186534] 
\draw  [color={rgb, 255:red, 16; green, 73; blue, 140 }  ,draw opacity=1 ][fill={rgb, 255:red, 16; green, 73; blue, 140 }  ,fill opacity=1 ] (97.4,134.46) .. controls (97.4,126.26) and (104.05,119.61) .. (112.26,119.6) .. controls (120.46,119.6) and (127.11,126.25) .. (127.11,134.45) .. controls (127.11,142.66) and (120.46,149.31) .. (112.26,149.31) .. controls (104.06,149.31) and (97.41,142.66) .. (97.4,134.46) -- cycle ;
%Shape: Circle [id:dp8986404213444509] 
\draw  [color={rgb, 255:red, 16; green, 73; blue, 140 }  ,draw opacity=1 ][fill={rgb, 255:red, 16; green, 73; blue, 140 }  ,fill opacity=1 ] (97.4,173.46) .. controls (97.4,165.26) and (104.05,158.61) .. (112.26,158.6) .. controls (120.46,158.6) and (127.11,165.25) .. (127.11,173.45) .. controls (127.11,181.66) and (120.46,188.31) .. (112.26,188.31) .. controls (104.06,188.31) and (97.41,181.66) .. (97.4,173.46) -- cycle ;

% Text Node
\draw (417,67) node [anchor=north west][inner sep=0.75pt]  [font=\Large] [align=left] {input-port qubit};
% Text Node
\draw (417,108) node [anchor=north west][inner sep=0.75pt]  [font=\Large] [align=left] {exit-port qubit};
% Text Node
\draw (417,147) node [anchor=north west][inner sep=0.75pt]  [font=\Large] [align=left] {switching qubit};

\end{tikzpicture}
    }
    \caption{\footnotesize Binary star resource with two switching qubits, shown as grey vertices. This resource is still blocking, as it can realize only one input-output connection per operational round.}
    \hrulefill
    \label{fig:02}
\end{figure}
%-------
\begin{lemma}[Blocking Condition]
\label{lem:01}
Let $I$ and $O$ be the sets of input- and output-port qubits of an $N$-crossbar resource $\ket{G}$, with $|I| = |O| = N$. 
The crossbar resource $\ket{G}$ is blocking if there exists a targeted graph $G'=(V',E')$ with $E'\subset I\times O$ and $|E'|<N$ such that, for every measurement strategy $\gamma(\cdot)$ that achieves $G'$, it holds that $\Gamma(G')\cap \left((I\cup O)\setminus V'\right)\neq \emptyset $.
\end{lemma}
\begin{IEEEproof}
Please refer to App.~\ref{app:l}
\end{IEEEproof}
Lemma~\ref{lem:01} gives an operational characterization of blocking in the quantum setting. Specifically, if a forwarding configuration, represented by a targeted graph $G'$, can be achieved only by measuring at least one port qubit that is not part of the active configuration, then such a port is depleted even though it should remain idle. This violates the non-blocking requirement, according to which any idle input-port must remain connectable to any idle output-port, without disturbing the connections already established within the fabric. As a consequence, achieving non-blocking operation in an entanglement-based switching fabric requires that all ports that are not currently engaged in a specific forwarding configuration remain available for subsequent entanglement forwarding. Equivalently, a non-blocking fabric must support the extraction of Bell states for any admissible combination of input- and output-ports, while preserving the reconfigurability of the residual resource. Hence, the parameter $N$ introduced in Def.~\ref{def:05} denotes not only the number of input- and output-ports, but also the maximum number of input-output Bell pairs that can be simultaneously extracted from the crossbar resource $\ket{G}$. This capability depends on the structure of the underlying resource state.

\begin{remark}[Switching qubits]
\label{rem:switching}
Lemma~\ref{lem:01} shows that non-blocking operation cannot be guaranteed if the resource state involves only port qubits. Indeed, port qubits must remain available as endpoints of input-output entanglement links, and therefore they cannot also serve as expendable degrees of freedom for reconfiguring the fabric. This motivates the introduction of an additional class of qubits, termed \textit{switching qubits}, whose role is to mediate entanglement forwarding between input- and output-ports while enabling a non-blocking operational mode.
\end{remark}
By introducing the switching qubits\footnote{Notably, the EC crossbar resource is a two-colorable state with one color including the port qubits and the other including the switching qubits. It is useful to recall that any graph state can be transformed into a two-colorable state with an overhead of additional qubits, which scales at most linearly \cite{CheCacCal-25}.}, the graph associated with the resource state is redefined as $G = (V,E)$, where the augmented vertex set $V = I \cup O \cup S$ collects the disjoint sets $I$ and $O$ of input- and output-port qubits, together with the set $S$ of \textit{switching qubits}. The edge set $E$ captures the entanglement connections among these vertices. However, the mere presence of switching qubits is not sufficient to guarantee the non-blocking condition. For instance, the binary star state~\cite{CheIllCac-24}, illustrated in Fig.~\ref{fig:02}, includes \textit{switching qubits} but still constitutes a \textit{blocking} crossbar resource. We formalize the conditions required for a non-blocking operation in the following design principle.
%-----
\begin{des}[Non-Blocking Entanglement-Based Switching Fabrics]
\label{def:06}
Consider a graph state $\ket{G}$ serving as an $N$-port-pair switching-fabric resource with vertex partition $V = I \cup O \cup S$,  where $I$ and $O$ are the disjoint input- and output-port sets ($I \cap O = \emptyset$), and $S$ is the switching-qubit set disjoint from the ports ($S \cap (I \cup O) = \emptyset$).  
The resource $\ket{G}$ realizes a \textit{non-blocking} $N$-fabric if the following conditions hold:
\begin{itemize}
\item[I)] For any forwarding configuration, namely matching $\mathcal{M}\subset I\times O$ with $|\mathcal{M}|<N$, there exists a measurement strategy $\gamma(\cdot)$ that achieves $\mathcal{M}$ by acting exclusively on switching qubits in $S$, such that
\begin{equation}
\label{eq:07}
\mathcal{P}^V_{\gamma}\ket{G}
\;\overset{\mathrm{LU}}{=}\;
\ket{G'} \otimes \ket{G''},
\end{equation}
where $\ket{G'}$ encodes the requested configuration $\mathcal{M}$, and $\ket{G''}$ is the residual graph state on the idle subsystem $V\setminus (V'\cup \Gamma(G'))$.
\item[II)] The residual graph state $\ket{G''}$ is itself a \textit{non-blocking} $(N-|\mathcal{M}|)$-fabric, i.e., it satisfies this design principle recursively on the remaining idle input- and output-ports.
\end{itemize}
\end{des}
Condition I ensures that any forwarding configuration $\mathcal{M}$ that does not saturate the fabric capacity ($|\mathcal{M}|<N$) can be realized without depleting port qubits. Hence, all idle ports remain available after the configuration is established. Condition II strengthens this requirement by imposing the same property recursively on the residual resource: after any partial configuration, the unmeasured subsystem must still constitute a non-blocking fabric for the remaining idle ports. 
Together, these conditions capture the quantum counterpart of non-blocking operation.
\begin{figure}
    \centering
    \resizebox{0.65\columnwidth}{!}{
    \input{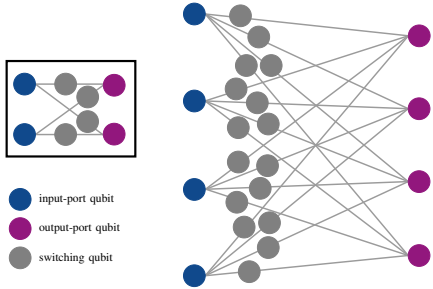}}
    \caption{\footnotesize Monolithic EC crossbar for $N=2$ and $N=4$. }
    \hrulefill
    \label{fig:03}
\end{figure}
Building on this non-blocking tenet, we introduce the \textit{monolithic} \textit{Edge-Controlled} (EC) crossbar, namely an entanglement-based switching fabric that satisfies the above conditions. In this design, a dedicated switching qubit is associated with each input-output pair. More precisely, the EC resource $\ket{G_{\text{EC}}}$ is defined as the graph state obtained by replacing each edge of the complete bipartite graph $K_{N}$ with a two-hop path mediated by a switching qubit, as formalized below and illustrated in Fig.~\ref{fig:03}.
\begin{resm}
\label{res:1}
Let $K_N=(V_K,E_K)$ be the complete bipartite graph with vertex partition 
$V_K = I \cup O$, where $I=\{i_1,\dots,i_N\}$ and $O=\{o_1,\dots,o_N\}$ are the disjoint sets of input- and output-port qubits, respectively, and $ 
E_K = \{\, (i,o), \, \forall i\in I,\; o\in O \,\}=I\times O$.
The \textit{Edge-Controlled (EC) resource state} $\ket{G_{\mathrm{EC}}}$ is the graph state associated with the graph 
$G_{\mathrm{EC}}=(V,E)$ obtained by \textit{enriching} $K_N$ through the introduction of a dedicated switching qubit $s_{i,o}$ for each edge $(i,o)\in E_K$.  
Each new vertex $s_{i,o}$ has exactly two adjacent vertices, namely the endpoints $i$ and $o$, $Adj(s_{i,o})=\{i,o\}$, thereby replacing each direct edge $(i,o)$ of $K_N$ with a two-hop path traversing $s_{i,o}$:
\begin{align}
\label{eq:EC}
V &\eqdef I \cup O \cup S,\\
S &\eqdef \{\, s_{i,o} \;|\; (i,o)\in E_K \,\text{\rm with}\,\, Adj(s_{i,o})=\{i,o\} \},\\
E &\eqdef \{\, (i,s_{i,o}), (s_{i,o},o) \;|\; (i,o)\in E_K \,\}.
\label{eq:ec-3}
\end{align}
\end{resm}
It is clear that each switching qubit controls the activation or inhibition of a direct entanglement link between the corresponding input qubit $i$ and output qubit $o$, through local Pauli measurements. 
The EC structure induces a natural partition of the switching-qubit set $S$ once a matching is realized. Specifically, given a matching $E'\subset I\times O$ as in Def.~\ref{def:04}, it is possible to identify the sets of active input- and output-ports as follows:
\begin{align}
    I' \eqdef \{\, i\in I \mid \exists!\,o\in O:\ (i,o)\in E'\,\}\\
    O' \eqdef \{\, o\in O \mid \exists!\,i\in I:\ (i,o)\in E'\,\}.
\end{align}
 Stemming from these sets, the switching set $S$ can be partitioned as follows:
\begin{itemize}
    \item{\textbf{Mediators:}} the switching qubits involved in activating the matching $E'$:
    \begin{equation}
    \label{eq:11}
    S_{\mathrm{med}} 
    \;\eqdef\; 
    \bigl\{\, s_{i,o}\in S \;\big|\; (i,o)\in 
     E'\}.
    \end{equation}\item{\textbf{Competitors:}} the switching qubits associated with edges connecting an active port either to idle ports or to other active ports not belonging to the current matching $E'$: 
    \begin{equation}
    \label{eq:12}
    S_\mathrm{com} 
    \;\eqdef\; 
    \bigl\{\, s_{i,o}\in S \;\big|\; (i,o)\in 
    \big[(I'\times O)\;\cup\;
    \bigl(I \times O'\bigr)\big] \setminus E'\}.
    \end{equation}
    \item{\textbf{Unused:}} the switching qubits associated with edges connecting idle ports, i.e.: \begin{equation}
    \label{eq:13}
    S_\mathrm{idle} 
    \;\eqdef\; 
    \bigl\{\, s_{i,o}\in S \;\big|i \notin I' \wedge o \notin O'\}.
    \end{equation}
\end{itemize}
Accordingly, mediators determine the direct connections as specified by $E'$, competitors represent potentially conflicting links incident to active ports, and the unused switching qubits correspond to the idle portion of the crossbar. In the following proposition, we prove that the EC resource is non-blocking. 
\begin{prop}[Non-Blocking EC Crossbar]
\label{prop:01}
Let $\ket{G_{\text{EC}}}$ be a $N$-EC crossbar resource. For any matching $\mathcal{M}\subset I\times O$, the EC resource is non-blocking according to Design Principle~\ref{def:06}.
\end{prop}

In essence, the non-blocking property arises from two \textit{invariants} of the proposed design: \textit{the port-safety} and \textit{the disjoint-switching-domain} properties. Each switching qubit mediates an independent input–output connection, ensuring that partial measurements never deplete the port qubits nor disturb unconnected edges. Consequently, the residual subsystem consistently preserves the same EC structure, thereby maintaining the non-blocking operational mode recursively.
Therefore, new direct entanglement links can always be established on top of an existing entanglement forwarding configuration by performing Pauli measurements only on idle switching qubits, while port qubits remain untouched. 

%----------------------------------
%---------Sec. IV------------------
%----------------------------------
\section{Modular Entanglement-based Switching Fabric Design}
\label{sec:4}

As detailed in Sec.~\ref{sec:5}, the monolithic EC crossbar requires $N^2$ switching qubits. This quadratic scaling reflects the price of a single-stage non-blocking design: the architecture is structurally simple and operationally direct, but its resource overhead grows with the square of the port number. 
Building upon the EC design, we now introduce a modular alternative that preserves the same non-blocking guarantee while reducing the switching-qubit count for sufficiently large $N$.

\subsection{Clos-Type Entanglement-Based Fabric}
\label{sec:4.1}

In classical switching theory, the Clos architecture provides a systematic way to realize large non-blocking switching fabrics by interconnecting smaller switching modules across multiple stages, thereby avoiding the quadratic growth of crosspoints inherent to monolithic crossbars. A three-stage Clos topology interconnects $N$ inputs to $N$ outputs through three stages of switches, referred to as ingress, middle, and egress stages. As represented in Fig.~\ref{fig:04}, this topology is parameterized by $(r,n,m) \in \mathbb{N}^3$, with $N=rn$. Specifically, it consists of an ingress stage composed of $r=\frac{N}{n}$ switches of size $n\times m$, a middle stage composed of $m$ switches of size $r\times r$, and an egress stage composed of $r=N/n$ switches of size $m\times n$. Each ingress switch is connected to all middle-stage switches, and each middle-stage switch is connected to all egress switches. If the number of middle-stage switches satisfies the condition $m\geq 2n-1$, then the classical three-stage Clos architecture is strictly non-blocking \cite{Clos-53}.
We adopt this architectural blueprint to move from monolithic EC fabric to a modular Clos-type entanglement-based switching fabric, where each classical switching module is replaced by a monolithic EC resource. Specifically, every classical $(n \times m)$, $(r \times r)$, and $(m \times n)$ switch in the ingress, middle, and egress stages, respectively, is substituted by an EC resource with the corresponding number of input- and output-port qubits.

\begin{resmd}
\label{def:clos_resource}
Let $r,n,m\in\mathbb{N}$, with $N=rn$, and consider a three-stage Clos-type architecture composed of: (i) $r$ ingress-stage modules, (ii) $m$ middle-stage modules, and (iii) $r$ egress-stage modules. Each module is an EC resource according to~\eqref{eq:EC}--\eqref{eq:ec-3}. We define the stage index sets as:
\begin{equation}
    \mathcal{L}\eqdef\{1,\dots,r\},\,\, 
    \mathcal{C}\eqdef\{1,\dots,m\},\,\,
    \mathcal{R}\eqdef\{1,\dots,r\}.
\end{equation}
A generic ingress-stage module is denoted by $L_\ell$ with $\ell\in\mathcal{L}$, a middle-stage module by $C_j$ with $j\in\mathcal{C}$, and an egress-stage module by $R_\rho$ with $\rho\in\mathcal{R}$. Accordingly, we define the module sets as:
\begin{align}
     L\eqdef\{L_\ell\}_{\ell\in\mathcal{L}},\,\, C\eqdef\{C_j\}_{j\in\mathcal{C}},\,\, R\eqdef\{R_\rho\}_{\rho\in\mathcal{R}},  
\end{align}
and the collection of all modules as:
\begin{align}
    &W=L\cup C \cup R.  
\end{align}
According to the Clos parametrization, each ingress-stage module $L_\ell \in L$ has $n$ external input-port qubits and $m$ inter-stage output-port qubits:
\begin{align}
I_{L_\ell} &\eqdef \{\, i_{\ell,a} \mid a\in\{1,\dots,n\}\,\},\,
O_{L_\ell} \eqdef \{\, o_{\ell,j} \mid j\in\mathcal{C}\,\}.
\label{eq:clos_labeling}
\end{align}
In \eqref{eq:clos_labeling}, the index $j$ used to label $o_{\ell,j}$ is inherited from the Clos interconnection pattern: the $j$-th output-port of $L_\ell$ is connected to the middle-stage module $C_j$. Thus, no additional local port index is needed for $O_{L_\ell}$, since its ports are naturally indexed by the middle-stage modules they connect to. Similarly, each middle-stage module $C_j$, has $r$ inter-stage input-port qubits and $r$ inter-stage output-port qubits:
\begin{align}
I_{C_j} &\eqdef \{\, i_{j,\ell} \mid \ell\in\mathcal{L}\,\},\,
O_{C_j} \eqdef \{\, o_{j,\rho} \mid \rho\in\mathcal{R}\,\}.
\end{align}
The label $\ell$ in $i_{j,\ell}$ univocally identifies the ingress module $L_\ell$ connected to that port, while the label $\rho$ in $o_{j,\rho}$ identifies the egress module $R_\rho$. Each egress-stage module $R_\rho$ has $m$ inter-stage input-port qubits and $n$ external output-port qubits:
\begin{align}
I_{R_\rho} \eqdef \{\, i_{\rho,j} \mid j\in\mathcal{C}\,\},
O_{R_\rho} \eqdef \{\, o_{\rho,a} \mid a\in\{1,\dots,n\}\,\},
\end{align}
where label $j$ in $i_{\rho,j}$ identifies univocally the middle module $C_j$ connected to that port. Each module $B\in W$ is endowed with its own set of intra-module switching-qubits $S_B$ and its own intra-module edge set $E_B$ according to the EC construction~\eqref{eq:EC}--\eqref{eq:ec-3}:
\begin{align}
&S_B \eqdef \{\, s^{(B)}_{u,v}\mid u\in I_B,\ v\in O_B\},
\, Adj\!\left(s^{(B)}_{u,v}\right)=\{u,v\} \\&
E_B
\eqdef
\{\, (u,s^{(B)}_{u,v}),\ (s^{(B)}_{u,v},v)\mid u\in I_B,\ v\in O_B \,\}.
\end{align}
The Clos interconnection is defined by the following inter-module edges. For every $(\ell,j)\in\mathcal{L}\times\mathcal{C}$, the output-port $o_{\ell,j}\in O_{L_\ell}$ is connected to the input-port $i_{j,\ell}\in I_{C_j}$:
\begin{equation}
\label{eq:LC}
E^{LC}\eqdef \{\, (o_{\ell,j},\, i_{j,\ell}) \mid \ell\in\mathcal{L},\ j\in\mathcal{C}\,\}.
\end{equation}
Similarly, for every $(j,\rho)\in\mathcal{C}\times\mathcal{R}$, the output-port $o_{j,\rho}\in O_{C_j}$ is connected to the input-port $i_{\rho,j}\in I_{R_\rho}$:
\begin{equation}
\label{eq:CR}
E^{CR}\eqdef \{\, (o_{j,\rho},\, i_{\rho,j}) \mid j\in\mathcal{C},\ \rho\in\mathcal{R}\,\}.
\end{equation}
Accordingly, the global external input- and output-port sets of the Clos-type fabric are defined as:
\begin{align}
I &\eqdef \bigcup_{\ell\in\mathcal{L}} I_{L_\ell},\,
O \eqdef \bigcup_{\rho\in\mathcal{R}} O_{R_\rho}.
\end{align}
The internal Clos ports are collected in the set
\begin{equation}
P\eqdef \big ( \bigcup_{B \in L \cup C} O_B ) \cup \big ( \bigcup_{B \in C \cup R} I_B \big).
\label{eq:clos_internal_ports}
\end{equation}
These vertices play a dual role. Locally, they are input- or output-port qubits of the corresponding EC modules. Globally, however, they act as switching qubits of the Clos-type fabric, since they mediate entanglement forwarding across different stages and enable the interconnection of distinct EC modules. Therefore, the global switching-qubit set of the Clos-type fabric includes both the internal Clos ports in $P$ and the intra-module switching qubits of all EC modules:
\begin{equation}
S \eqdef P
\cup \left(\bigcup_{B\in W} S_B\right).
\label{eq:clos_switching_set}
\end{equation}
The Clos-type EC graph is then $G=(V,E)$, with
\begin{align}
V &\eqdef 
I\cup O\cup S,
\label{eq:clos-1}\\
E &\eqdef 
\left(\bigcup_{B\in W} E_B\right)
\cup E^{LC}\cup E^{CR}.
\label{eq:clos-3}
\end{align}
The \textit{Clos-type EC resource} is the graph state $\ket{G}$ associated with the graph $G=(V,E)$.
\end{resmd}

%---------
\begin{figure}
    \centering
    \resizebox{0.8\columnwidth}{!}{
    \input{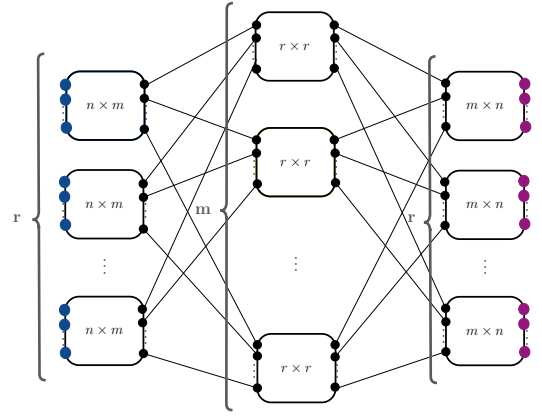}
    }
    \caption{\footnotesize Schematic representation of a three-stage Clos topology with stages $L, C, R$. The ingress stage $L$ and the egress stage $R$ contain $r$ modules of size $n\times m$ and $m\times n$, respectively. The middle stage $C$ contains $m$ modules of size $r\times r$. Overall, the topology realizes a switching fabric with $N=rn$ external inputs and $N=rn$ external outputs.}
    \hrulefill
    \label{fig:04}
\end{figure}
%-------
Fig.~\ref{fig:05} illustrates a Clos-type EC fabric with $N=6$, $r=3$, $n=2$, and $m=3$. The ingress input-ports and egress output-ports form the external interface, while the remaining qubits form the global switching set $S= P \cup \left(\bigcup_{B\in W} S_B\right)$. The intra-module qubits $\bigcup_{B\in W} S_B$ mediate forwarding within each EC module, whereas the internal Clos ports $P$ mediate forwarding across consecutive stages. 

Let $E' \subset I \times O$ be a matching between the external input- and output-ports of the Clos-type fabric. Similarly to the monolithic case, we define the active external port sets as:
\begin{align}
I' &\eqdef \{\, i\in I  \mid \exists!\, o\in O :\ (i,o)\in E' \,\},\\
O' &\eqdef \{\, o\in O  \mid \exists!\, i\in I :\ (i,o)\in E' \,\}.
\end{align}
As discussed in Sec.~\ref{sec:3}, in the monolithic EC crossbar, the matching $E'$ directly induces a partition of the switching-qubits: each pair $(i,o)\in E'$ identifies its unique dedicated mediator $s_{i,o}$, and the corresponding competitor and unused sets follow without ambiguity. In the Clos-type fabric, this one-to-one correspondence breaks down: the same matching $E'$ can be realized through multiple internal paths, depending on which middle-stage module is selected for each connection. Each such selection induces a different partition of the global switching set $S$, and thus a different measurement strategy. Hence, the internal path selection must be made explicit, as formalized by the following definition. 

\begin{definition}[Middle-Stage Assignment]
    \label{def:06}
    Let $E'\subset I\times O$ be a matching between the external input- and output-port qubits of a Clos-type EC fabric with stage index sets $\mathcal{L}$, $\mathcal{C}$, $\mathcal{R}$. A \textit{middle-stage assignment} for $E'$ is a function $\Lambda: E'\rightarrow\mathcal{C}$, which assigns to each requested input-output pair $(i,o)\in E'$ a middle-stage module index $\Lambda(i,o)\in\mathcal{C}$, such that, for any two distinct pairs  $(i,o),(i',o')\in E'$, the following holds:
    \begin{align}
    \label{eq:routing}
    \nonumber
    \Big(\exists &\ell\in\mathcal{L}: i,i'\in I_{L_\ell} \Big)\  \vee\ \Big( \exists\,\rho\in\mathcal{R}: o,o'\in O_{R_\rho} \Big) \\& \Longrightarrow \quad \Lambda(i,o)\neq \Lambda(i',o').
    \end{align}
\end{definition}

Condition~\eqref{eq:routing} prevents contention on inter-stage links. Indeed, two connections originating from the same ingress module $L_\ell$ cannot use the same middle-stage module $C_j$, since $L_\ell$ has a unique inter-stage port toward $C_j$. Similarly, two connections terminating at the same egress module $R_\rho$ cannot use the same middle-stage module $C_j$, since $R_\rho$ has a unique inter-stage port connected from $C_j$.

\begin{remark}
    The Clos-type EC fabric mirrors exactly the classical SNB Clos structure. Indeed, the internal-path selection algorithm is responsible only for selecting the middle-stage module used to realize each connection in the matching. But it does not specify the Pauli-measurement strategy or any other quantum operation. Accordingly, any classical routing algorithm for SNB Clos architectures can be used as the internal-path selection algorithm for the Clos-type EC fabric.
\end{remark} 
\noindent \paragraph*{Assignment-induced local matchings}
The assignment $\Lambda$ maps the external matching $E'$ into three local matchings, one for each stage of the Clos-type fabric. Indeed, let us consider a pair $(i,o) \in E'$ with $i\in I_{L_\ell}$ and $o\in O_{R_\rho}$ and let $\Lambda(i,o)=j$. The selected middle-stage module $C_j$ induces the following local input-output pairs:
\begin{itemize}
    \item in the ingress module $L_\ell$: $(i,\; o_{\ell,j})\in I_{L_\ell}\times O_{L_\ell}$,
    \item in the middle module $C_j$: $(i_{j,\ell},\; o_{j,\rho})\in I_{C_j}\times O_{C_j}$,
    \item in the egress module $R_\rho$: $(i_{\rho,j},\; o)\in I_{R_\rho}\times O_{R_\rho}$.
\end{itemize}
Indeed, once $C_j$ is selected, the Clos interconnection uniquely identifies the output port $o_{\ell,j}$ of $L_\ell$ connected to the input port $i_{j,\ell}$ of $C_j$, as specified by~\eqref{eq:LC}. Similarly, the output port $o$ determines the egress module $R_\rho$, and hence the output port $o_{j,\rho}$ of $C_j$ connected to the input port $i_{\rho,j}$ of $R_\rho$, as specified by~\eqref{eq:CR}. By applying the above reasoning to each input-output pair in the matching, the global matching $E'$ together with the middle-stage assignment $\Lambda$ induces, for each module $B\in W$, a local EC matching $E'_B\subset I_B\times O_B$:
\begin{align}
E'_{L_\ell} 
&\eqdef 
\{\, (i,o_{\ell,j}) \mid (i,o)\in E',\ i\in I_{L_\ell},\ \Lambda(i,o)=j \,\},\\
E'_{C_j} 
&\eqdef 
\{\, (i_{j,\ell},o_{j,\rho}) \mid (i,o)\in E',\ i\in I_{L_\ell},\ o\in O_{R_\rho},\nonumber\\
&\quad \quad \quad \quad \quad \quad  \quad  \quad \quad \quad \quad \quad \quad \quad \,\,\, \Lambda(i,o)=j \,\},\\
E'_{R_\rho} 
&\eqdef 
\{\, (i_{\rho,j},o) \mid (i,o)\in E', o\in O_{R_\rho},\ \Lambda(i,o)=j\}.
\end{align}
The admissibility condition in~\eqref{eq:routing} guarantees that each $E'_B$ is a valid matching within the corresponding EC module, since no two local pairs share the same local input or output port.
\begin{figure}
    \centering
    \resizebox{0.8\columnwidth}{!}{
    \input{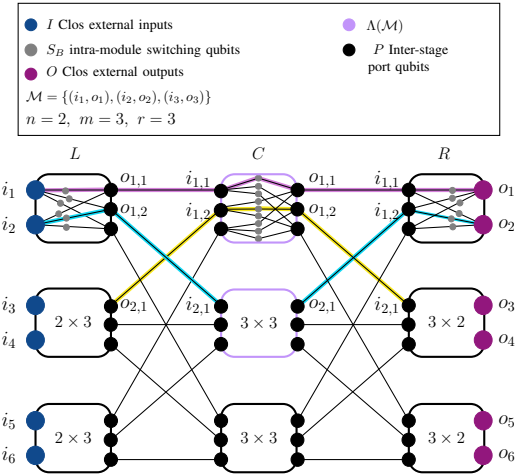}
    }
    \caption{\footnotesize Modular EC fabric with $N=6,r=3,n=2,m=3$.}
    \label{fig:05}
    \hrulefill
\end{figure}
With these definitions, we can now express the partition of the Clos-type EC fabric qubits induced by a given matching $E'$ and middle-stage assignment $\Lambda$.
\paragraph*{Active inter-stage ports}
For each module $B\in W$, we define the sets of active inter-stage input- and output-ports as:
\begin{align}
I'_B 
&\eqdef 
\{\, u\in I_B \mid \exists!\, v\in O_B:\ (u,v)\in E'_B\,\},\\
O'_B 
&\eqdef 
\{\, v\in O_B \mid \exists!\, u\in I_B:\ (u,v)\in E'_B\,\}.
\end{align}
Then, the set of assignment-induced mediator port-qubits is the subset of active ports that are internal to the fabric:
\begin{equation}
P_{\mathrm{med}}
\eqdef
\Bigl( \bigcup_{B\in W} (I'_B \cup O'_B) \Bigr)\cap P.
\end{equation}
Thus, the matching $E'$ and the assignment $\Lambda$ induce a partition of the global switching-qubit set $S \eqdef P
\cup \left(\bigcup_{B\in W} S_B\right)$ into mediators, competitors, and idle qubits as follows.
\begin{itemize}
\item  \textbf{Mediators:}
\begin{align}
&S^{(B)}_{\mathrm{med}}
\eqdef
\bigl\{\, s^{(B)}_{u,v}\in S_B : (u,v)\in E'_B \,\bigr\},\\&
S_{\mathrm{med}}
\eqdef
P_{\mathrm{med}}
\ \cup\
\bigcup_{B\in W} S^{(B)}_{\mathrm{med}}.
\end{align}
\item \textbf{Competitors:}
\begin{align}
S^{(B)}_{\mathrm{com}}
&\eqdef
\Bigl\{ s^{(B)}_{u,v}\in S_B:(u,v)\in
\nonumber\\
&\quad
\bigl[(I'_B\times O_B)\cup(I_B\times O'_B)\bigr]
\setminus E'_B
\,\Bigr\},\\
S_{\mathrm{com}}
&\eqdef
\bigcup_{B\in W} S^{(B)}_{\mathrm{com}}.
\end{align}
There are no competitor inter-stage ports: the admissibility condition in~\eqref{eq:routing} ensures that no two routed connections contend for the same internal Clos port.

\item \textbf{Idle qubits:}
\begin{align}
S^{(B)}_{\mathrm{idle}}
&\eqdef
\bigl\{\, s^{(B)}_{u,v}\in S_B :
u\notin I'_B \ \wedge\ v\notin O'_B
\bigr\},\\
P_{\mathrm{idle}}
&\eqdef
P\setminus P_{\mathrm{med}}, \quad
S_{\mathrm{idle}}
\eqdef P_{\mathrm{idle}}  \cup \,
\bigcup_{B\in W} S^{(B)}_{\mathrm{idle}}.
\end{align}
    
\end{itemize}
Thus, a matching $E'$ together with a middle-stage assignment induces a well-defined partition of the Clos-type fabric qubits into mediator, competitor, and idle qubits. Mediators belong to the internal paths realizing the requested external input-output connections. Competitors correspond to intra-module links incident to active local ports but not selected by the local matchings. Idle qubits represent the unused portion of the fabric. 
In the following, we prove that the modular EC fabric satisfies the Design Principle~\ref{def:06}.

\begin{prop}[Non-Blocking Clos-type EC fabric]
\label{cor:02}
Let $\ket{G}$ denote a three-stage Clos-type EC resource with $(r,n,m)\in\mathbb{N}^3$ and $N = rn$. If $m\geq 2n-1$, i.e., the classical strictly non-blocking (SNB) Clos condition holds, then the modular Clos-type resource $\ket{G}$ realizes a non-blocking $N$-port-pair switching fabric according to Design Principle~\ref{def:06}. 
\end{prop}
\begin{IEEEproof}
    Please refer to App.~\ref{app:03}
\end{IEEEproof}
The above non-blocking property follows from two cooperating mechanisms. At the module level, each EC subgraph independently satisfies the port-safety and disjoint-switching-domain invariants established in Prop.~\ref{prop:01}: no intra-module forwarding operation ever depletes a port qubit. At the fabric level, the SNB condition $m\geq 2n-1$ guarantees that, as long as idle external ports remain, at least one middle-stage module is always available to mediate the forwarding. Together, these two properties ensure reconfigurability, where every partial forwarding configuration leaves a residual graph state that is itself a valid non-blocking fabric on the remaining idle ports, satisfying Design Principle~\ref{def:06} recursively.
\begin{remark}
    The non-blocking property of monolithic and modular resources is structural: it is determined by the structure of the underlying graph itself and can therefore be certified a priori, independently of any runtime reconfiguration or routing strategy.
\end{remark}
%----------------------------
%------------Sec. V-----------
\section{Analysis and Comparison}
\label{sec:5}
We compare the proposed EC switching fabrics with three representative EPR-based architectures. These baselines are not chosen arbitrarily. Rather, they capture three recurrent ways of realizing input-output entanglement forwarding with bipartite resources: (i) a fully pre-shared bipartite architecture, where one EPR pair is available for every possible input-output pair; (ii) a star architecture, where all requested connections contend for a shared BSM resource as in \cite{IneAlvCho-24}; and (iii) a memory-relay architecture, in which quantum-memory banks are interconnected through a switchboard as in~\cite{LeeBerDah-22}.

This comparison also mirrors the classical taxonomy of switching fabrics. The star architecture is the quantum analogue of switching via a shared bus: a centralized BSM resource serializes the service of multiple input-output requests. The memory-relay router resembles switching via memory: forwarding is mediated by memory banks and an active internal switchboard, which must select, pair, and process independently generated bipartite resources before the requested input-output entanglement is completed. In contrast, the proposed EC fabrics implement the quantum analogue of switching via a crossbar interconnection network: independent input-output connections are activated in parallel through measurements on a pre-generated multipartite resource. The purpose of the comparison is therefore to isolate the architectural distinction between EPR-based and multipartite-based forwarding.

%-------------
\subsection{Forwarding Latency}
\label{sec:5.1}
To characterize the architectural contribution to latency, we distinguish between two complementary quantities. The \textit{forwarding depth} $D(\mathcal{M})$ is the number of sequential operational rounds required to realize $\mathcal{M}$. $D(\mathcal{M})$ captures the intrinsic degree of serialization imposed by the forwarding architecture. The \textit{forwarding latency} $\bar{T}(\mathcal{M})$ is the expected time to complete all forwarding operations for $\mathcal{M}$, and additionally captures the impact of probabilistic primitives. 
Formally, let $\mathcal{M}\subseteq I\times O$ be a matching with $|\mathcal{M}|\leq N$, as in Def.~\ref{def:04}. We define $\bar{T}(\mathcal{M})$ as the expected time required to establish all $|\mathcal{M}|$ Bell pairs specified by $\mathcal{M}$, starting from the instant at which the matching is assigned to the switching fabric. Let $\delta$ denote the number of available Bell-state-measurement (BSM) devices in EPR-based architectures, and let $\sigma$ denote the number of available Pauli-measurement apparatuses in multipartite-resource architectures. The case $\delta=1$ ($\sigma=1$) corresponds to a single measurement device unit and gives the \textit{sequential depth}, representing the worst case over parallel configurations, while with $\delta$($\sigma$) parallel units, the depth reduces to the parallel depth. The resulting architectural comparison is summarized in Tab.~\ref{tab:latency_scaling}. We now analyze each architecture in turn.

\begin{figure}
   \centering
    \resizebox{0.85\columnwidth}{!}{
    \input{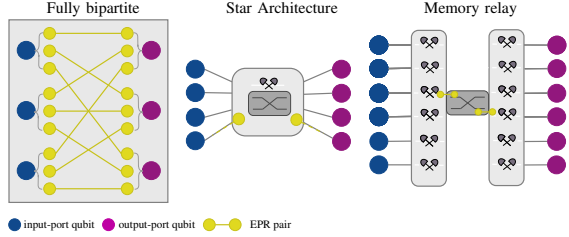}
    }
    \caption{\footnotesize EPR-based switching fabric architectures.}
    \label{fig:M}
    \hrulefill
\end{figure}

\begin{table*}[t]
\centering
\caption{Forwarding-latency comparison grouped by matching-awareness regime.}
\label{tab:latency_scaling}
\begin{tabularx}{\textwidth}{l c c c X}
\toprule
\textbf{Architecture} 
& \textbf{Regime}
& \textbf{Sequential depth}
& \textbf{Parallel depth}
& \textbf{Expected forwarding latency} \\
\midrule

Fully bipartite
& Matching-oblivious
& $0$
& $0$
& $\approx 0$ \\[4pt]

EC fabric
& Matching-oblivious
& --
& $\left\lceil N_{\mathrm{meas}}(\mathcal M)/\sigma \right\rceil$;
  $1$ if $\sigma\geq N_{\mathrm{meas}}(\mathcal M)$
& $O\!\left(\left\lceil N_{\mathrm{meas}}(\mathcal M)/\sigma \right\rceil(\tau_q+\tau_c)\right)$;
  $\tau_q+\tau_c$ if $\sigma\geq N_{\mathrm{meas}}(\mathcal M)$ \\[10pt]

\midrule

Star architecture
& Matching-driven
& $|\mathcal M|$
& $\left\lceil |\mathcal M|/\delta \right\rceil$
& $O\!\left(\left\lceil |\mathcal M|/\delta \right\rceil / p_{\mathrm{BSM}}\right)$ \\[8pt]

Memory-relay router
& Matching-driven
& $|\mathcal M|$
& $\left\lceil |\mathcal M|/\delta \right\rceil$
& $O\!\left(\left\lceil |\mathcal M|/\delta \right\rceil /(p_{\mathrm{gen}}p_{\mathrm{BSM}}^{2})\right)$ \\
\bottomrule
\end{tabularx}
\end{table*}

\paragraph*{Fully bipartite architecture}
In this architecture, one Bell pair is pre-established for every input-output pair. Under the assumption that all Bell pairs required by $\mathcal{M}$ are already available, forwarding does not require any additional switching operation. Hence, $D(\mathcal{M})=0$ and the corresponding forwarding latency is negligible. This regime, however, is achieved at the price of a quadratic memory overhead: each input port must be associated with a bank of $N$ memory qubits, so that the architecture maintains $N^2$ pre-shared Bell pairs. This resource cost is discussed further in Sec.~\ref{sec:5.2}.

\paragraph*{Star architecture}
In this architecture, all input-output connections are realized through a shared BSM resource. In the single-device case, i.e., $\delta=1$, each pair in $\mathcal{M}$ must access the swapping resource sequentially, yielding $D_{\mathrm{star}}(\mathcal{M}) = |\mathcal{M}|.$
% \begin{equation}
% D_{\mathrm{star}}(\mathcal{M}) = |\mathcal{M}|.
% \end{equation}
With $\delta$ parallel BSM devices the forwarding depth reduces to $D_{\mathrm{star}}(\mathcal{M}) = \lceil |\mathcal{M}|/\delta \rceil$. Since each BSM attempt succeeds independently with probability $p_{\mathrm{BSM}}$, the expected number of attempts per connection is $1/p_{\mathrm{BSM}}$, and the $|\mathcal{M}|$ connections are served one at a time. Thus, the expected forwarding latency scales linearly with the matching size:
\begin{equation}
\bar{T}_{\mathrm{star}}(\mathcal{M}) \propto \left\lceil \frac{|\mathcal{M}|}{\delta} \right\rceil
\frac{1}{p_{\mathrm{BSM}}},
\end{equation}
up to a constant factor associated with the duration of a single BSM attempt and the related classical feed-forward. In the single-device case, this reduces to $\bar{T}_{\mathrm{star}}(\mathcal{M})
\propto
\frac{|\mathcal{M}|}{p_{\mathrm{BSM}}}$.

\paragraph*{Memory-relay router}
In this architecture, each requested connection requires the successful generation of an entangled pair between the memory banks, followed by BSMs implementing entanglement swapping across the memories. In the single-device case, i.e., $\delta=1$, generation and swapping are performed sequentially for the pairs in $\mathcal{M}$, yielding a forwarding depth that grows linearly with the matching size: $D_{\mathrm{QR}}(\mathcal{M})=|\mathcal{M}|$. With  $\delta$ parallel internal service chains, the forwarding depth reduces to
\begin{equation}
D_{\mathrm{QR}}(\mathcal{M})
=
\left\lceil \frac{|\mathcal{M}|}{\delta} \right\rceil.
\end{equation}
Under the simplifying assumption that each connection requires one successful internal entanglement-generation event with probability $p_{\mathrm{gen}}$ and two successful BSMs, each with probability $p_{\mathrm{BSM}}$, the expected forwarding latency scales as:
\begin{equation}
\bar{T}_{\mathrm{QR}}(\mathcal{M})
\propto
\left\lceil \frac{|\mathcal{M}|}{\delta} \right\rceil
\frac{1}{p_{\mathrm{gen}}\,p_{\mathrm{BSM}}^{2}}.
\end{equation}
In the single-device case, this becomes $\bar{T}_{\mathrm{QR}}(\mathcal{M})=\frac{|\mathcal{M}|}{p_{\mathrm{gen}}\,p_{\mathrm{BSM}}^{2}}$.

\paragraph*{EC switching fabrics}
In both the monolithic EC and the Clos-type EC fabrics, forwarding is realized through local Pauli measurements on the switching qubits. Since distinct matched edges are associated with disjoint switching-qubit domains, the corresponding measurements commute and can be executed in parallel. Accordingly, let $N_{\mathrm{meas}}(\mathcal{M})\leq |S|$ denote the total number of Pauli measurements required to implement $\mathcal{M}$. With $\sigma$ measurement apparatuses available, the forwarding depth is $D_{\mathrm{EC}}(\mathcal{M})= \lceil N_{\mathrm{meas}}(\mathcal{M})/\sigma\rceil$. In the regime $\sigma \geq N_{\mathrm{meas}}(\mathcal{M})$, all measurements can be executed simultaneously, reducing the depth to a single operational round, and yielding:
\begin{equation}
D_{\mathrm{EC}}(\mathcal{M}) = 1, \quad \bar{T}_{\mathrm{EC}}(\mathcal{M}) = \tau_q + \tau_c,
\end{equation}
where $\tau_q$ is the duration of a single-qubit Pauli measurement and $\tau_c$ accounts for classical signaling and conditional local corrections.
The comparison reveals that the fundamental distinction is not whether the resource is bipartite or multipartite per se, but whether the architecture is universal for all admissible matchings or must react to the specific matching request. The fully bipartite architecture and the proposed EC fabrics belong to a \textit{matching-oblivious forwarding} regime: the provisioned resource is able to support any admissible matching before the matching itself is assigned. In the fully bipartite case, universality is achieved by storing one Bell pair for every possible input-output pair. In the EC case, universality is encoded structurally in a pre-generated multipartite graph-state fabric, and the matching only determines the local Pauli measurements to be performed. By contrast, the star architecture and the memory-relay router belong to a \textit{matching-driven forwarding} regime: the requested matching must be known before the architecture can select, steer, generate, or swap the required bipartite resources. This induces a per-connection service bottleneck and leads to matching-dependent forwarding depth.

\subsection{Resource Scaling}
\label{sec:5.2}
We now compare the resource required to instantiate each forwarding architecture. Specifically, we analyze the number of qubits that must be stored or made available to support forwarding.

\paragraph*{Star architecture}
A star architecture stores one qubit per port and relies on a shared BSM resource to realize the requested input-output connections. Its qubit requirement is linear, i.e., $O(N)$. This reduced requirement in qubit number comes at the cost of a matching-dependent latency: with a single BSM the requested connections are served sequentially, whereas with $\delta$ parallel BSM devices the depth reduces to $\lceil\frac{|\mathcal{M}|}{\delta}\rceil$, as discussed in Sec.~\ref{sec:5.1}. 

\paragraph*{Fully bipartite architecture}
A direct pre-sharing of Bell pairs between every input-output pair requires $N^2$ Bell pairs distributed over $2N$ input-output-ports memory banks for a total of $2N^2$ qubits. This quadratic provisioning makes the architecture matching-oblivious in terms of resource availability: once the Bell pairs are available, any admissible matching can be served with negligible forwarding latency by selecting the corresponding pre-shared pairs. However, this universality is obtained by full pairwise replication: each logical input-port is associated with $N$ distinct physical qubits, one for each possible output connection. Hence, the forwarding decision, i.e. the matching $\mathcal{M}$ to be served in a given round, selects not only the Bell pair to be consumed but also which physical qubit realizes the logical port for that connection. This couples port identity, memory lookup, and forwarding choice, a coupling that disappears in the EC-Crossbar.

\paragraph*{Monolithic EC crossbar}
For the monolithic EC $N$-crossbar, one switching qubit is dedicated to each input-output pair. Thus $|S|_{\mathrm{EC}} = N^2 $ and $ |V|= N^2+2N$. The total qubit count is therefore quadratic, as in the fully bipartite architecture. The two designs differ, however, in how they achieve universality over all admissible matchings. The fully bipartite architecture stores $N^2$ independent Bell pairs, one for each possible input-output connection. The EC crossbar encodes universality in a single multipartite graph-state fabric: the same structured resource supports all admissible matchings, and a specific matching is realized by choosing the appropriate local Pauli measurements on the switching qubits. Hence, the EC crossbar removes the coupling highlighted for fully bipartite architectures. Indeed, each logical port corresponds to a unique physical port-qubit, independently of the selected output. The matching determines which switching qubits are measured, not which physical qubit realizes the port.  This separation is architecturally important: the external port interface remains fixed, while forwarding decisions are implemented by configuring the internal degrees of freedom of the fabric. Thus, the EC crossbar replaces a replicated bank of pair-labeled port memories with a measurement-programmable switching substrate, providing a clear boundary between routing decisions, port interfacing, and internal forwarding operations. In other words, the fully bipartite architecture makes the physical realization of a port matching-dependent, whereas the EC crossbar makes only the internal fabric configuration matching-dependent.

\paragraph*{Modular EC fabric}
The total switching-qubit count required by the modular Clos EC fabric is obtained by summing the intra-module contributions and the internal port qubits $P$:
\begin{equation}
\label{eq:clos_s_count}
|S|_{\mathrm{Clos}} = \underbrace{\vphantom{\frac{m}{n^2}}2Nm + mr^2}_{\mathrm\sum_{B}|S_B|}
+\underbrace{\frac{4mN}{n}}_{\mathrm|P|}
=\frac{m}{n^2}N^2+\frac{2m(n+2)}{n}N
\end{equation}
where $r=N/n$, with graph-state size $|V|_{\mathrm{Clos}} = 2N+|S|_{\mathrm{Clos}}$.

The modular and monolithic EC fabrics have quadratic graph-state size. However, under the strictly non-blocking Clos condition $m=2n-1$, the leading coefficient of the Clos-type fabric is $\frac{m}{n^2}= \frac{2n-1}{n^2} < 1$ for $ n\geq 2$. Since both EC architectures envision $2N$ external port qubits, the crossover is equivalently obtained from $|S|_{\mathrm{Clos}}=|S|_{\mathrm{EC}}$, yielding:
\begin{equation}
N^{*}(n) = \frac{2(2n-1)n(n+2)}{(n-1)^2}.
\end{equation}
For $n=2,3,4,5$, this gives $N^*\approx 48,38,37,39$, respectively. Thus, for $N>N^*(n)$, the Clos-type EC fabric reduces the graph-state size with respect to the monolithic EC crossbar, as illustrated in Fig.~\ref{fig:S_scaling}. The corresponding asymptotic leading coefficients are $3/4,\,5/9,\,7/16,\,9/25$. Beyond the reduced qubit-count, the modularity of Clos-type fabric also provides a structural advantage: its largest individual EC module contains at most $\max\!\left\{n(2n-1),\,\frac{N^2}{n^2}\right\}$ intra-module switching qubits, instead of the $N^2$ switching qubits of the monolithic one. 
\begin{figure}[t]
\centering
\includegraphics[width=\columnwidth]{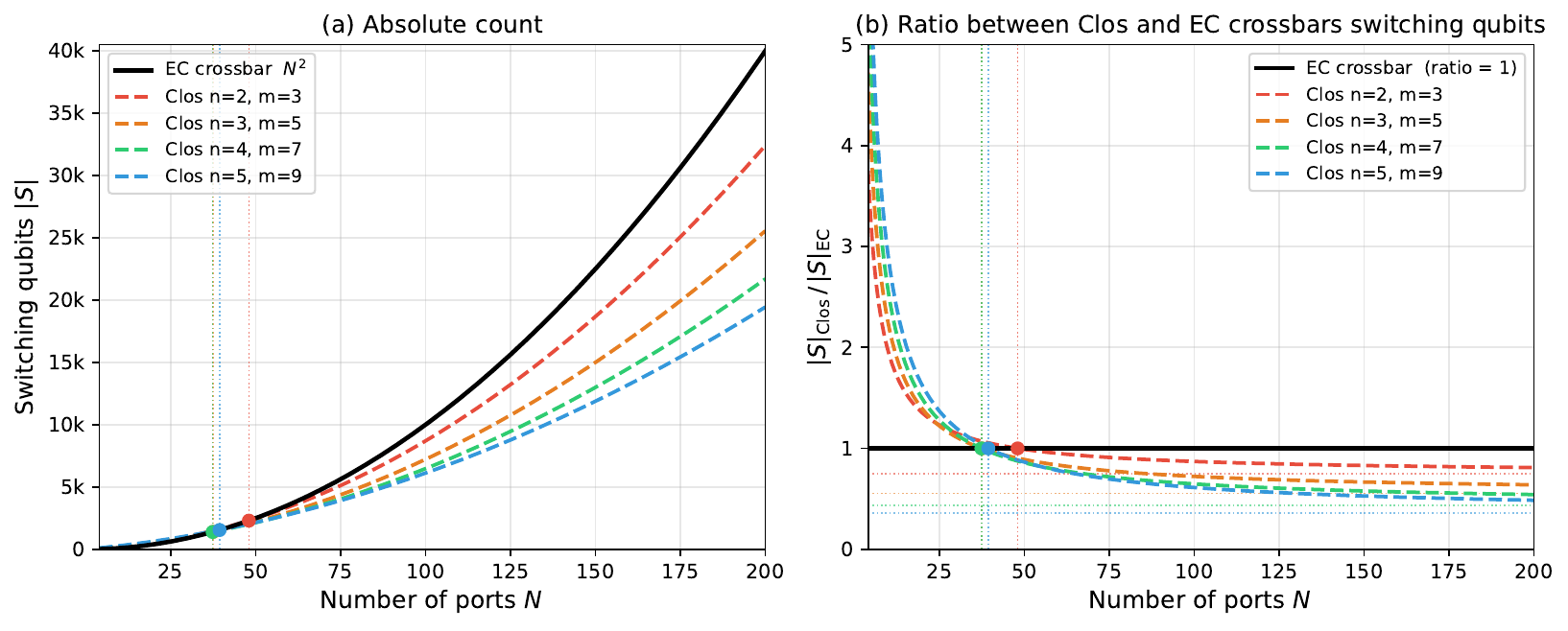}
\caption{\footnotesize Switching-qubit scaling vs the number of ports $N$. (a) Absolute count for the monolithic EC crossbar and Clos-type EC fabrics with $m=2n-1$; dots mark the crossover $N^*$ where the Clos-type count becomes smaller. (b) Ratio $|S|_{\mathrm{Clos}}/|S|_{\mathrm{EC}}$, dotted lines indicate the asymptotic factors $m/n^2$.}
\label{fig:S_scaling}
\hrulefill
\end{figure}

\subsection{Implementation considerations}
\label{sec:5.3}
The hardware required to realize entanglement forwarding differs substantially between bipartite and multipartite architectures, both in the nature of the forwarding primitive and in the role of the supporting devices.

In photonic implementations a BSM typically relies on two-photon interference followed by coincidence detection, i.e., a \textit{double-click} event. Under linear optics, the success probability is bounded by $p_{\mathrm{BSM}} \leq 1/2$~\cite{CasLut-01}. Recent experiments beyond standard linear-optical schemes have achieved higher success probabilities, up to $62.5\%$~\cite{BayDarVan-23}, at the price of increased experimental complexity and the use of ancillary qubits. Hence, in EPR-based architectures, forwarding depends on probabilistic primitives whose success probability directly affects the expected latency.

In these architectures, an internal switchboard is required to route qubits from the channels or memories to the BSM station. In the \textit{star architecture}, the switchboard collects photons associated with the requested input-output pairs and routes them to a shared BSM device. With a single BSM resource, the switchboard acts as a time-multiplexed scheduler: only one pair can be served at a time, leading to the matching-dependent latency characterized in Sec.~\ref{sec:5.1}. Parallel BSM devices reduce this serialization but do not remove the dependence on the number of requested pairs.
In the \textit{memory-relay router}, the switchboard has an additional role. Besides routing memory qubits to the measurement stations, it must coordinate the internal entanglement resources required to complete the requested input-output connections. Under the abstraction used in Sec.~\ref{sec:5.1}, each connection requires one successful internal entanglement-generation event with probability $p_{\mathrm{gen}}$ and two successful BSMs, each with probability $p_{\mathrm{BSM}}$. This compounds the probabilistic nature of the forwarding process and yields the latency scaling reported in Tab.~\ref{tab:latency_scaling}.

The EC crossbar replaces the BSM-based forwarding pipeline with a measurement-driven paradigm closely related to \textit{measurement-based quantum computing}~(MBQC)~\cite{RauBri-01}. The forwarding operation is not realized by routing qubits to shared BSM stations, but by choosing the measurement pattern applied to a pre-generated multipartite fabric. Hence, the primary implementation cost is shifted to the generation of $\ket{G}$. Note that entangled states of this class can be produced by deterministic protocols based on sequential photon emission from quantum emitters~\cite{SchCogScm-16,RubMasYel-25}, including one- and two-dimensional cluster states with high fidelity and without post-selection. Furthermore, the pro-active generation represents a widely recognized approach, as remarked in \cite{AbaCubMai-25,CheCacCal-25}. Once $\ket{G}$ is available, forwarding reduces to local single-qubit Pauli measurements, implemented through basis selection followed by detection. This primitive is structurally different and simpler than a BSM: it does not require two-photon interference between the qubits being connected, and it can be parallelized across the switching qubits involved in the selected matching. Overall, BSM-based designs require an active switchboard that routes qubits to shared measurement resources and, in memory-relay architectures, coordinates internal entanglement generation and swapping. Their forwarding latency is therefore tied to the number of requested connections and to the success probabilities of the underlying probabilistic primitives. EC fabrics, instead, move the dominant hardware cost to the preparation of the graph-state resource. Once this resource is available, the forwarding stage is reduced to parallel local measurements, achieving depth-one operation. Thus, the advantage of the EC crossbar lies not only in forwarding latency, but also in the architectural simplification of the active forwarding hardware: the role of the switchboard is replaced by a measurement-programmable entanglement fabric that decouples port interfacing from internal forwarding logic. This advantage is not captured by the asymptotic qubit count alone: it concerns the control abstraction exposed by the forwarding substrate. The fully bipartite design exposes pair-labelled Bell-pair availability to the forwarding controller, whereas the EC fabric exposes a stable port interface and confines matching-dependent decisions to internal measurement choices.
%-----------------------------------
%---------Sec. VI. -----------------
%-----------------------------------

\section{Conclusions}
\label{sec:6}

This paper introduced a switching-fabric framework for quantum routers based on multipartite entanglement. By translating the classical notion of non-blocking operation into the quantum setting, we derived the edge-controlled (EC) design principle and instantiated it through two architectures: a monolithic EC crossbar and a modular Clos-type EC fabric. The former provides a direct non-blocking construction, while the latter imports the modular decomposition of classical Clos networks to reduce the switching-qubit overhead in the appropriate parameter regime. The proposed fabrics are \textit{matching-oblivious}: their forwarding resources are provisioned to support all admissible matchings before the matching itself is assigned. In EC fabrics, this universality is encoded structurally in a multipartite graph-state resource.  Moreover, EC fabrics decouple forwarding decisions from port interfacing: each logical port is represented by a stable physical port qubit, while matching-dependent operations are confined to the internal switching degrees of freedom of the fabric. This decoupling provides a clean architectural boundary between routing decisions, port interfaces, and forwarding operations. The framework is hardware-agnostic and should be understood as an architectural abstraction rather than as a device-specific implementation. Its main cost is the generation and maintenance of the multipartite resource state, while the forwarding stage is reduced to local single-qubit measurements. In turn, the separation between resource preparation and forwarding logic enables the switching-fabric architecture to be studied independently of the physical mechanisms used to generate the underlying graph state.

%--------------------------------------------------------------
% Appendices
%--------------------------------------------------------------
\begin{appendices}
\section{Proof of Lemma~\ref{lem:01}}
\label{app:l}

Let $G'=(V',E')$, with $|V'|$ \textit{active} input-output-ports, be the achievable graph. Since $|E'|<N$, at least one input and one output-port remain \textit{idle} under the target configuration $G'$, i.e., $V\setminus V'\neq \emptyset$.
By hypothesis, for every measurement strategy $\gamma(\cdot)$ that achieves $G'$, there exists at least one port $i_p \in \Gamma(G') \cap (V\setminus V')$
(or equivalently $o_p$, since the argument applies symmetrically for $i_p\in I$ or $o_p\in O$). Hence, although $i_p$ is an idle port, it is nonetheless measured, and thus depleted, in the process of realizing $G'$, being removed from the original graph. This, in turn, implies that $i_p$ can no longer be used to establish additional entanglement-based connections by local operations and classical communication (LOCC). 
Let $o_q$ denote an idle port on the opposite side of the switching fabric (if $i_p\in I$, then $o_q\in O\setminus V'$ and vice-versa). The pair $(i_p,o_q)$ constitutes an \textit{idle} input-output-port pair that cannot be subsequently connected, since $i_p$ is no longer available. This directly violates the non-blocking criterion, which requires that any idle input can be connected to any idle output without disturbing existing connections.
Hence, the switching-fabric resource $\ket{G}$ is \textit{blocking}. 

\section{Proof of Proposition~\ref{prop:01}}
\label{app:01}
According to Design Principle~\ref{def:06}, an entanglement-based crossbar is non-blocking if for any matching, there exists a measurement strategy resulting in \eqref{eq:07} and the residual state is still a non-blocking crossbar resource. The proof proceeds in steps: consider the matching associated with $E'\subset I\times O$ as introduced above, together with the three EC-induced subsets of switching qubits
$S_{\mathrm{idle}}$, $ S_{\mathrm{com}}$, and $S_{\mathrm{med}}$.\\
\noindent\textbf{Step 1 (Per-edge strategy and port safety).} Let us consider the following measurement strategy:
\begin{equation}
\label{eq:10}
    \gamma(s_{i,o})= \begin{cases}\sigma_y & s_{i,o} \in S_{\mathrm{med}}\\
    \sigma_z & s_{i,o} \in  S_{\mathrm{com}}\\
    \end{cases}
\end{equation}
According to this strategy, for each $(i,o)\in E'$, a Pauli-$Y$ is performed on the corresponding mediator $s_{i,o}\in S_{\mathrm{med}}$, while Pauli-$Z$ measurements are applied to the competitors $S_{\mathrm{com}}$, and nothing is done on the other switching qubits. 
A Pauli-$Y$ measurement on $s_{i,o}$ enables (up to local unitaries on the neighbors) a direct edge between its adjacent vertices $i$ and $o$, i.e.,  $(i,o)$, and removes $s_{i,o}$ from the graph~\cite{HeiEisBri-04,HeiDurEis-06,MazCalCac-24-1,CheIllCac-24}. Similarly, the Pauli-$Z$ measurements delete all the competitor switching qubits and all their incident edges, without altering the rest of the graph. 
We note that each switching qubit $s_{i,o}$ is uniquely associated with its pair $(i,o)$, thus the measurements specified by $\gamma(\cdot)$ act on pairwise disjoint sets of switching qubits, i.e., on $S_{\mathrm{med}} \cup S_{\mathrm{com}}$, with $S_{\mathrm{med}}\cap  S_{\mathrm{com}}=\emptyset$. This in turn implies that port qubits are never measured, that is, \textit{ports are never depleted}:
\begin{equation}
\label{eq:15bis}
    \Gamma(G')=S_{\mathrm{med}}\cup S_{\mathrm{com}} \subseteq S\,\, \wedge\, \Gamma(G')\cap (I\cup O)= \emptyset.
\end{equation}
\noindent\textbf{Step 2 (Commutation and independence).}
The EC structure ensures that all switching qubits in $S_{\mathrm{med}} \cup S_{\mathrm{com}}$ are distinct. Thus, the corresponding single-qubit Pauli operators commute, since they act on disjoint subsets of qubits. 
Hence, the measurement strategy $\gamma(\cdot)$, performed on $S_{\mathrm{med}} \cup S_{\mathrm{com}}$, yields the same post-measurement state, regardless of order:
\begin{equation}
\label{eq:16}
\mathcal P^V_{\gamma}\ket{G_{\mathrm{EC}}}
\;\overset{\mathrm{LU}}{=}\;
\ket{G'}\otimes \ket{\phi_{\mathrm{idle}}},
\end{equation}
where $\ket{G'}=\bigotimes_{(i,o)\in E'} \ket{\Phi_{io}}$ is the achieved graph with edge set $E'$, and $\ket{\phi_{\mathrm{idle}}}$ is the post-measurement state on the remaining subsystem $V\setminus (V'\cup\Gamma(G'))$ with $\Gamma(G')=S_{\mathrm{med}} \cup S_{\mathrm{com}}$. \\
\noindent\textbf{Step 3 (Structure of the residual and recursion).}
As proved above, the depleted set is $\Gamma(G')$ as expressed in \eqref{eq:15bis}.
% \begin{equation}
%     \Gamma(G)=S_{\text{med}}(E')\cup S_{\mathcal C}(E') \subseteq S
% \end{equation}
% and no port qubit is measured. 
Therefore, the residual subsystem is:
\begin{align} 
V_{\mathrm{idle}}\;=V\;\setminus (V'\cup \Gamma(G'))& \;=\; \\
\Big((I\setminus I') \cup (O\setminus O')\Big)\ \cup\ \Big( S \setminus \big (S_{\mathrm{com}}&\cup S_{\mathrm{med}}\big)\Big).\nonumber
\end{align}

For any two \textit{idle} ports $i\in I\setminus I'$ and $o\in O\setminus O'$, the switching qubit $s_{i,o}\in S_{\mathrm{idle}}$ remains present and adjacent only to \textit{both} endpoints, as it is unaffected by the measurement strategy $\gamma(\cdot)$.  Hence, the EC structure is preserved among the idle ports. In particular, the residual subsystem $\ket{\phi_{\mathrm{idle}}}$ is, up to local unitaries, equivalent to $\ket{G''_{EC}}$, with vertex set $(I\setminus I')\,\cup \,(O\setminus O')\,\cup \,(S\setminus \Gamma(G'))$ and edge set $\{(i,s_{i,o}),(s_{i,o},o)\;|\; i \notin I' \wedge o \notin O'\}$, which reproduces the $(N-|E'|)$-EC crossbar structure in \eqref{eq:EC}-\eqref{eq:ec-3}. Overall, \eqref{eq:16} can be re-written as:
\begin{equation}
    \mathcal P^V_{\gamma}\ket{G_{\mathrm{EC}}}
\;\overset{\mathrm{LU}}{=}\;
\bigotimes_{(i,o)\in E'} \ket{\Phi_{io}}
\ \otimes\
\ket{G''_{\mathrm{EC}}}.
\end{equation}

\noindent\textbf{Step 4 (Non-blocking property).}
By repeating Steps~1–3 on $\ket{G''_{\mathrm{EC}}}$ for any further matching among the idle ports, the same construction applies inductively. 
Therefore the residual $\ket{G''_{\mathrm{EC}}}$ is itself a non-blocking $(N-|E'|)$-crossbar in the sense of Design Principle~\ref{def:06}.
This completes the proof by establishing that $\ket{G_{\mathrm{EC}}}$ is a non-blocking monolithic EC crossbar, and therefore a valid non-blocking switching-fabric resource.

\section{Proof of Proposition~\ref{cor:02}}
\label{app:03}

We fix a matching $\mathcal{M}\subset I\times O$ with $|\mathcal{M}|<N$ and prove Design Principle~\ref{def:06}.%:

\textit{Middle-stage assignment under the SNB condition.}
Under the strictly non-blocking Clos condition $m\ge 2n-1$, it is always possible to assign, for each $(i,o)\in\mathcal{M}$, a middle-stage module that is not already used by any connection incident to the same ingress- or egress-stage module. This follows from the definition of SNB Clos architecture \cite{Clos-53}.
Accordingly, there exists an assignment $\Lambda$ defined on every pair $(i,o)\in \mathcal{M}$ such that the induced internal paths do not collide on inter-stage port qubits.
Fix such an assignment $\Lambda$ and consider the associated assignment-induced local matchings $E'_B$ with the corresponding partition of the fabric qubits into
$S_{\mathrm{med}}$,
$S_{\mathrm{com}}$,
and
$S_{\mathrm{idle}}$,
as defined in Sec.~\ref{sec:4.1}.

\textit{The Clos-type EC fabric satisfies Design Principle ~\ref{def:06}}
Consider the measurement strategy \eqref{eq:10}, with $S_{\mathrm{med}}$,
$S_{\mathrm{com}}$, defined as in Sec.~\ref{sec:4.1}.
% Define the measurement strategy $\gamma(\mathcal{M})$ as follows:
% \begin{equation}
% \label{eq:clos_strategy}
%     \gamma(v)= \begin{cases}
%         \sigma_y & v \in S_{\mathrm{med}} \\
%         \sigma_z & v \in S_{\mathrm{com}} 
%     \end{cases}
% \end{equation}
By construction, $\gamma(\cdot)$ acts exclusively on qubits in $S$ and therefore $\Gamma(\mathcal{M})\subset S$ $\Gamma(\mathcal{M})\cap (I\cup O)=\emptyset$.

By the assignment-induced partition of the modular EC fabric as in Sec.~\ref{sec:4.1}, the switching-qubit set is decomposed as $ S= S_{\mathrm{med}} \cup S_{\mathrm{com}}  \cup S_{\mathrm{idle}} , $ with $
    S_{\mathrm{med}}\cap S_{\mathrm{com}}=S_{\mathrm{com}}\cap S_{\mathrm{idle}}= S_{\mathrm{med}} \cap S_{\mathrm{idle}}= \emptyset .$
% \begin{align}
%     S= S_{\mathrm{med}} \cup &S_{\mathrm{com}}  \cup S_{\mathrm{idle}} , \nonumber\\
%     S_{\mathrm{med}}\cap S_{\mathrm{com}}=S_{\mathrm{com}}\cap &S_{\mathrm{idle}}= S_{\mathrm{med}} \cap S_{\mathrm{idle}}= \emptyset .
% \end{align}
Since single-qubit Pauli operators acting on different qubits commute, the global operator induced by $\gamma(\cdot)$ in \eqref{eq:10} factorizes as follows:
\begin{align}
    \mathcal{P}^V_\gamma= \big( \bigotimes_{v \in S_{\mathrm{com}}} \sigma_z^{(v)}\big)\big( \bigotimes_{v \in S_{\mathrm{med}}} \sigma_y^{(v)}\big)
\end{align}
and the post-measurement state depends on the set of measured qubits, not on the order in which measurements are performed. 
%Equivalently, all Pauli-$Y$ and Pauli-$Z$ measurements can be carried out in any order or simultaneously.

%\paragraph*{Proof of Requirement~I}
Consider an arbitrary pair $(i,o)\in\mathcal{M}$ and its selected internal path under $\Lambda$.
For each module traversed by the path, all switching qubits incident to the corresponding active ports except the one selected by the local matching belong to $S_{\mathrm{com}}$ and are measured in the $Z$ basis.
As a result, all alternative internal links, i.e. intra-module connections to active internal ports, incident to active ports are removed. The remaining mediator switching qubits inside the EC modules are measured in the $Y$ basis. As in the monolithic case, each such measurement consumes a degree-two switching qubit and, up to local Pauli corrections, induces an effective edge between the corresponding pair of ports selected by the local matching.

Due to the $Z$-measurements, each inter-stage port qubit selected by the assignment is left adjacent only to the two qubits belonging to the routed path. Hence, these inter-stage port qubits are degree-two vertices in the reduced graph. Measuring them in the $Y$ basis thus contracts the routed internal path and, up to local unitaries, yields a Bell pair $\ket{\Phi_{io}}$ between the external input $i$ and the external output $o$.

We observe that the classical SNB routing guarantees that the sequence of qubits measured to connect two pairs of the matching are disjoint. Furthermore, after competitor removal, the neighborhood of the mediator qubits is also disjoint. Thus, the Pauli-$Y$ measurement associated with a given pair does not affect the graph transformation induced by another pair of the matching. Repeating the competitor removal and mediator measurement for all $(i,o)\in\mathcal{M}$, we obtain $\mathcal{P}^V_{\gamma}\ket{G} \;\overset{\mathrm{LU}}{=}\;\Bigl(\bigotimes_{(i,o)\in\mathcal{M}}\ket{\Phi_{io}}\Bigr)\otimes \ket{\phi_{idle}},$
% \begin{equation}
%     \mathcal{P}^V_{\gamma}\ket{G}
% \;\overset{\mathrm{LU}}{=}\;
% \Bigl(\bigotimes_{(i,o)\in\mathcal{M}}\ket{\Phi_{io}}\Bigr)\otimes \ket{\phi_{idle}},
% \end{equation}
where $\ket{\phi_{idle}}$ denotes the post-measurement state on the remaining subsystem $V\setminus (V'\cup \Gamma(\mathcal{M}))$.
This proves that the strategy $\gamma(\cdot)$ on $\ket{G}$ can achieve the matching $\mathcal{M}$.

We observe that if $|\mathcal{M}|< N$, $V_{\mathrm{idle}}\neq \emptyset$ and we can explicit the connection pattern among the residual qubits as follows:
\begin{align}
\label{eq:resiudal_graph-1}
    V'' &= V_{\mathrm{idle}}=(I\setminus I')\cup(O\setminus O')\cup S_{\textrm{idle}}\\S_{\textrm{idle}}&=P_{\mathrm{idle}}\cup \bigcup_{B\in W} S^{(B)}_{\mathrm{idle}} \\
E'' &= E''_w \cup \big(\bigcup_{B \in W }E''_B\big)
\\
E''_B &= \{(i,s_{i,o}),(s_{i,o},o)\;|\; i \notin I'_B \wedge o \notin O'_B\} \label{eq:64}\\
E''_w &= \bigl\{(p,q)\in E^{LC}\cup E^{CR}
        \mid p,q\in P_{\mathrm{idle}}\bigr\}
%E''_w= \{(o_{\ell,j}, i_{j,\ell}): \ell \in \mathcal{L}\setminus \mathcal{L_A}, j \in \mathcal{C}\setminus \mathcal{C_A} \}  \cup \nonumber\\
%\{(o_{j,\rho}, i_{\rho,j}): j \in \mathcal{C}\setminus \mathcal{C_A}, \rho \in \mathcal{R}\setminus \mathcal{R_A}, \}
\label{eq:resiudal_graph-3}
\end{align}
where $E''_B$ denotes the set of edges within EC modules and $E''_w$ denotes the set of inter-stage edges. 
By Clos construction and $m\geq 2n-1$, if $|\mathcal{M}|<N$ then $E''_w\neq \emptyset$. Indeed, if $|\mathcal{M}|<N$, then $V_{\mathrm{idle}} \neq \emptyset $, and the condition $m\geq 2n-1$ guarantees that there exists at least one middle stage able to connect idle inputs with idle outputs. That is, $P_{\mathrm{idle}}\neq \emptyset$ and hence $E''_w \neq \emptyset$. 
Hence, $\ket{\phi_{idle}}$ is a graph state $\ket{G''}$ corresponding to the graph $G''=(V'', E'')$, where $V''$ and $E''$ are described in \eqref{eq:resiudal_graph-1}-\eqref{eq:resiudal_graph-3}. Thus, the proof follows.

\end{appendices}

\bibliography{biblio}
\bibliographystyle{ieeetr}

\end{document}